\documentclass[review,authoryear]{article}
\pdfoutput=1
\usepackage{mathrsfs}
\usepackage{amssymb}
\usepackage{graphicx}
\usepackage{amsmath,amsfonts,amsthm}
\usepackage{algorithm}
\usepackage{algorithmic}
\usepackage{listings}
\usepackage{lineno}
\usepackage{bm}
\usepackage{xr}
\usepackage{natbib}
\usepackage{units}

\newtheorem{lemma}{Lemma}[section]
\newtheorem*{lemma*}{Lemma}
\newtheorem{theorem}{Theorem}[section]
\newtheorem*{theorem*}{Theorem}
\newtheorem{conjecture}{Conjecture}[section]

\newtheorem*{definition*}{Definition}

\newcommand{\be}{\begin{equation}} \newcommand{\ee}{\end{equation}}
\newcommand{\bd}{\begin{displaymath}} \newcommand{\ed}{\end{displaymath}}
\newcommand{\ba}{\begin{align}} \newcommand{\ea}{\end{align}}
\newcommand{\baa}{\begin{align*}} \newcommand{\eaa}{\end{align*}}
\newcommand{\ben}{\begin{enumerate}} \newcommand{\een}{\end{enumerate}}
\newcommand{\bi}{\begin{itemize}} \newcommand{\ei}{\end{itemize}}

\newcommand{\ud}{\mathrm{d}}
\newcommand{\E}[1]{\operatorname{E}\left[ #1 \right]}
\newcommand{\Var}[1]{\operatorname{Var}\left[ #1 \right]}
\newcommand{\var}[1]{\operatorname{Var}\left[ #1 \right]}
\newcommand{\cov}[2]{\operatorname{Cov}\left[ #1,#2 \right]}

\begin{document}

\title{Quantifying the effects of anagenetic and cladogenetic evolution}

\author{{\sc Krzysztof Bartoszek}}
\maketitle
\begin{abstract}
An ongoing debate in evolutionary biology is whether phenotypic change occurs predominantly around the time
of speciation or whether it instead accumulates gradually over time. In this work I propose a general framework
incorporating both types of change, quantify the effects of speciational change via
the correlation between species
and attribute the proportion of change to each type. I discuss
results of parameter estimation of Hominoid body size in this light. 
I derive mathematical formulae related to this problem, the
probability generating functions of the number of speciation
events along a randomly drawn lineage and from the most recent common ancestor
of two randomly chosen tip species for a conditioned Yule tree.
Additionally I obtain in closed form
the variance  of the distance from the root to the 
most recent common ancestor of two randomly chosen tip species.
\end{abstract}


\section{Introduction}
One of the major debates in evolutionary biology concerns when evolutionary change takes place. 
Two predominant theories state that changes take place either at times of speciation 
\citep[punctuated equilibrium][]{NEldSGou1972,SGouNEld1993}
or gradually over time \citep[phyletic gradualism, see references in][]{NEldSGou1972}. 

Phyletic gradualism describes how Darwin envisioned evolution happening  
\citep{NEldSGou1972}.
The theory of punctuated equilibrium was brought up in response to what
fossil data was indicating \citep{NEldSGou1972,SGouNEld1977,SGouNEld1993}.
One rarely observed a complete unbroken fossil series
but more often distinct forms separated by
long periods of stability \citep{NEldSGou1972}. 
Falconer and Darwin
at the birth of the theory of evolution 
already discussed this 
and the latter
saw ``the fossil record more 
as an embarrassment than as an aid to his theory'' \citep{NEldSGou1972}.

Therefore instead of continuous evolution \citet{NEldSGou1972} developed the idea
that change takes place predominantly at speciation
in the form of rapid evolution which are effectively jumps at the time 
scale of most phylogenies \citep{GSim1947}. 
They were of course not the first to do this. As already mentioned
such a theory can be found in Falconer's writings and was
presented under the name ``quantum evolution'' by \citet{GSim1947}.
In fact \citet{GSim1947} postulated that this quantum evolution
was responsible for the emergence of taxonomic units of high order.
However, the work of \citet{NEldSGou1972} re--introduced 
punctuated equilibrium into contemporary mainstream evolutionary theory.

One of the theoretical motivations behind punctuated equilibrium is
that the organism has to function well as a whole. Any particular change
to any trait can easily make it dysfunctional with the rest of the organism.
Therefore traits have to be adapted to each other first
\citep{SChe1961,TDob1956,TFra1975}. Any initial change to a trait 
will most probably be inadaptive or non--adaptive first
\citep{GSim1947}. One cannot expect the change to be large as then it could 
make the organism poorly adapted to its environment or the trait would not function well
in co--operation with other traits or functions of other body parts.
The probability of a mutation with a large
positive adaptive phenotypic effect occurring is very small \citep{GSim1947}.
However, it might happen that 
a new environmental niche appears \citep[due to some species dying out and leaving it empty][]{GSim1947} for which 
the direction of the change 
is beneficial. Therefore if the organism starts taking up the niche
and enhancing this new change, rapid evolution has to take place
for all body parts to catch up until a new equilibrium is reached. 
The alternative \citep[and usual][]{GSim1947} result of such a change is that the organism becomes extinct due
to no niche appearing 
while at the current one the new trait was at a disadvantage
or evolution was not rapid enough to fill up the new niche.


\citet{GSim1947} discussed that it would be easiest for 
such a change to fixate inside a ``small and completely isolated population.''
\citet{SChe1961} also postulated this 
mechanism for differentiation due to isolation inside a species.
\citet{EMay1982} hypothesized that genetic imbalance 
of isolated subpopulations of a species would favour change
and 
additionally pointed
out that this situation is a possible alternative explanation to the
observed long periods of stasis in the fossil record. 
Paleontologists usually encounter widespread populous species
and these by the previous argument would be the least likely 
to change. 

Showing whether a trait gradually evolves in a continuous manner
or is subject to large rapid changes followed by relative stability
is not only a result in itself but gives us actual insights about the
role of the trait for the organism.
\citet{AMooSVamDSch1999} discuss this.
They expect cladogenetic evolution to occur in
traits that are involved in speciation or niche shifts and gradual evolution
in traits under continuous selection pressures that often change direction.

Biological evidence for both gradual and punctuated evolution is present. 
As discussed by 
\citet{FBok2002} punctuated equilibrium is supported by fossil records \citep[see][]{NEldSGou1972,SGouNEld1977,SGouNEld1993}
and as already mentioned they were the motivation for the development of the theory.
On the other hand
\citet{TOza1975}'s study of the 
Permian verbeekinoid foraminifer \textit{Lepidolina multiseptata}
provides evidence for gradual evolution.
\citet{GStebFAya1981} in their work also point to 
experiments supporting phyletic gradualism. 
In a more recent work \citet{GHun2008} found that both a punctuated model and gradual model
seemed plausible for evolution
of pygidial morphology on a lineage of the trilobite \textit{Flexicalymene}.
One should suspect, as \citet{GSim1947} pointed out, that evolution will be a mix of different modes
and 
interestingly body size evolution in the radiolarian \textit{Pseudocubus vema}
and shape components in the foraminifera \textit{Globorotalia}
were best explained by a model with both anagenetic and cladogenetic components \citep{GHun2008}.
Even more recently \citet{GHun2013} studied the deep--sea ostracode genus \textit{Poseidonamicus}.
Using a maximum likelihood approach he found that only one or two shape traits seem to have evolved
via speciational change while four other shape traits evolved gradually (Brownian motion).

Mathematical models with gradual and punctuated components have been considered previously in the literature.
Examples of this are due to \citet{FBok2002,FBok2008,FBok2010,TMatFBok2008,AMooDSch1998,AMooSVamDSch1999}.
In a very recent work \citet{JEasDWegCLeuLHar2013} study \textit{Anolis} lizards evolution 
using a Brownian motion with jumps inside branches model.
Testing, especially based on extant species only, whether gradual or punctuated evolution is dominating
or if both contribute similarly can be a difficult matter.
One only observes the contemporary sample and it is necessary to divide the signal between the two
sources of stochasticity. 
Work has of course been done in this direction. For example
\citet{JAvi1977} studied the type of evolution in some contemporary fish families  and
\citet{TMatFBok2008} asked whether gradual or punctuated evolution is responsible for body size evolution in
mammals. \citet{GHun2008} developed a likelihood--based statistical framework 
for testing whether a purely gradual (Brownian motion) or a purely punctuated model
or a mix of these two mechanisms explained phenotypic data best. 
It should also be remembered that measurement error can 
make estimation and testing even more difficult.
Even with only gradual change, measurement error 
can have very complicated effects in comparative data e.g.
phylogenetic regression coefficients can be both upward or downward biased
depending on both the tree and true model parameters 
\citep{THanKBar2012}.

The evolutionary model considered in this work has two main components. 
One is the model
of phenotypic evolution and the other the branching process underlying the phylogeny. 
Currently most of the phylogenetic comparative literature is set in a framework
where one conditions on a known tree. Whilst with the current wealth
of genetic data which allows for more and more accurate phylogenies this 
is a logical framework, one can run into situations where 
this is not sufficient. Typical examples are fossil data or 
unresolved clades e.g. 
in the Carnivora order
\citep[used for an example analysis in][]{FCraMSuc2013}. 
To add to this, very recently a new member of the olingos (\textit{Bassaricyon}, order Carnivora)
the Olinguito (\textit{B. neblina}) has been described \citep{Olinguito}.
Moreover when considering models with a 
jump component one has to consider speciation
events leading to extinct lineages and how this interplays with 
the model of phenotypic evolution assumed here.
The tree and speciation events themselves might not be of direct interest, in fact they could actually be 
nuisance parameters. So there is interest in tree--free
methods that preserve distributional properties of the observed phenotypic values
\citep{FBok2010}.

I consider in this work two models of phenotypic evolution. Brownian
motion, interpreted as unconstrained neutral evolution 
and a single optimum Ornstein--Uhlenbeck model.
In the case of the second model the  phenotype tends to constrained oscillations around the attractor. 
On top of this just after a speciation
event each daughter lineage 
undergoes rapid change, described here by some particular probability density function.

Combining the Ornstein--Uhlenbeck process with a jump component
makes phylogenetic comparative models consistent with the
original motivation behind punctuated equilibrium.
After a change occurs rapid adaptation of an organism
to the new situation takes place followed by stasis. Stasis as underlined
by \citet{SGouNEld1993} does not mean that no change occurs. 
Rather that during it ``fluctuations of little or no accumulated consequence'' occur. 
The Ornstein--Uhlenbeck 
process fits into this idea. If the adaptation rate is large enough
then the process reaches stationarity very quickly and oscillates
around the optimal state and this can be interpreted as stasis
between the jumps --- the small fluctuations. 
\citet{EMay1982} supports this sort of reasoning by hypothesizing
that ``The further removed in time a species from the original
speciation event that originated it, the more its genotype
will have become stabilized and the more it is likely to resist change.''

Literature concerning punctuated equilibrium models is predominantly 
concentrated on applications and estimation.
Here I take a more theoretical approach and study the mathematical properties
of evolutionary models 
with  an additional component of rapid phenotypic change associated with speciation
assuming that all model parameters are known (or have been preestimated).
To the best of my knowledge an adaptive evolutionary model with a 
punctuated component has not been developed fully. 
The main contribution of this paper
is the mathematical development of such a model ---
the Ornstein--Uhlenbeck model with jumps.

Current Ornstein--Uhlenbeck based evolutionary models
have a discontinuity component but this is in the 
primary optimum \citep{KBaretal,MButAKinOUCH,THan1997}.
The phenotype evolves gradually towards oscillations around a new optimum value
and this can correspond to a change in the environment, a sudden replacement of the
adaptive niche. The model considered here is different. It is the 
phenotype which jumps, it is pushed away from its optimum 
due to e.g. a mutation that changed the character dramatically. 
After the jump the phenotype evolves back towards the optimum.

Of course it would be more realistic to have a discontinuity in 
both the phenotype and the primary optimum. Then when the jump in the phenotype
and optimum function would coincide the species would be at an advantage
needing a shorter time to reach equilibrium. However, this setup
would require some sort of model for the optimum and is beyond the scope
of this study.

The other model considered by us is a Brownian motion model with jumps.
It differs from the Ornstein--Uhlenbeck one in that it does not
have an adaptive constraint. In fact, as I show in my mathematical derivations,
it is the limit as the rate of adaptation goes to $0$. From a biological
point of view this model illustrates unconstrained evolution with rapid speciational change. 
Additionally it serves as a warm--up to the mathematics of the significantly more involved Ornstein--Uhlenbeck model.

The Ornstein--Uhlenbeck with jumps model combines three important evolutionary components. Firstly it has built into
it a phylogenetic inertia component, 
the ancestral value of the trait will make up part of the contemporary species's trait value 
and additionally cause dependencies between current species.  
Secondly unlike the Brownian motion model the variance of this process tends to a constant and 
so after some time the trait should exhibit constrained oscillations around the optimum value --- stasis.
And finally there is the jump component --- the possibility that at speciation the trait
undergoes rapid evolution.

In this work apart from
introducing a model combining
adaptation and punctuated equilibrium
I propose a way of quantitatively
assessing the effect of both types of evolution
so that one can consider a mix of gradual and punctuated components. 
The analytical results are 
computed for a pure birth model of tree growth, however using
example data I show that they can be carried over (at the moment via simulation
methods) to branching process models that include extinction.

The phylogeny, number and timing of speciation events, is modelled 
by a constant--rates birth--death process
conditioned, on the number of tip species, $n$.
By $\lambda$ I denote the birth rate 
and by $\mu$ the death rate. 
Contemporary literature on this is substantial
\citep[e.g.][]{DAldLPop2005,TGer2008a,AMooetal,TreeSim1,TreeSim2,TStaMSte2012}.
The key mathematical property of conditioned branching processes 
is that conditional on the tree's height the times of speciation
events are independent and identically distributed random variables.
The distribution is particular to the regime:
critical ($\lambda=\mu$),
supercritical ($\lambda > \mu >0$) or pure birth ($\lambda>\mu=0$).
Estimation of birth and death rates has been
widely discussed in the literature
\citep[e.g.][]{FBok2002,FBok2003,FBokVBriTSta2012,SNee2001,SNeeEHolRMayPHar1994,SNeeRMayPHar1994}.

The combination of evolutionary processes and conditioned branching
processes has already been considered by \citet{AEdw1970} who proposed
a joint maximum likelihood estimation procedure of a pure birth tree and a 
Brownian motion on top of it. Markov--Chain Monte Carlo based methods to jointly
estimate the phylogeny and parameters of a Brownian motion trait
have been proposed by \citet{JHueBRanJMas2000} and \citet{JHueBRan}.
\citet{GSlaetal2012} develop an Approximate Bayesian Computation framework to estimate
Brownian motion parameters in the case of an incomplete tree.
\citet{SSagKBar2012,FCraMSuc2013,KBarSSag2012} have contributed by
considering a Brownian motion on an unobserved birth--death tree in the first two works
and an Ornstein--Uhlenbeck process on an unobserved pure birth tree in the third.
These studies concentrate on the speciation process driving the phenotypic process,
as is usually assumed in phylogenetic comparative methods, where the trait is of main
interest. 
There have been a number of recent papers concerning models
where the rate of speciation depends
on the trait values.
\citet{MPieJWei2005} discuss many possible modelling examples
where the phenotypic and speciation processes interact. 
\citet{WMadPMidSOtt2007} derive likelihood formulae 
for a model where a binary character drives speciation and extinction
and study this model via simulations.
Closer to the present setting is \citet{RFit2010}
who assumed the trait evolved as a diffusive process 
with birth and death coefficients as functions of the trait. 
A similar framework of combining Ornstein--Uhlenbeck process
and branching processes has been considered by
\citet{ARosABraUDie2010} and by \citet{ABraLCarARos2011} in the food--web
structure community.

The organization of this paper is as following. In Section \ref{secYule} I discuss the pure birth model
of tree growth, derive the probability generating function of the number of speciation events on a random lineage
and from the most recent common ancestor of two randomly sampled tip species and in addition also show how my results can
be applied to working with the total tree area. Section \ref{secPEq} is devoted to punctuated equilibrium evolutionary models. 
Section \ref{secHom}
interprets the results of the Hominoidea analyses of \citet{FBok2002,FBokVBriTSta2012} in the light
of the present models and Section \ref{secDis} is a discussion,  Appendix A contains proofs 
of the main mathematical results Theorems \ref{thPGF} and \ref{thRhon}, and 
Appendix B presents a brief introduction to quadratic variation.

\section{The conditioned Yule process}
\label{secYule}
\subsection{Yule model of tree growth}
The Yule tree model \citep{GYul1924} is a pure birth Markov branching 
process. At the beginning there is one species that lives for an exponential
time and then splits into two species each behaving in the same manner. 
Here I will consider a Yule process conditioned on $n$ tips: the extant species
These types of branching processes, called  conditioned birth--death processes,
have received significant attention in the last decade 
\citep[e.g.][]{DAldLPop2005,TGer2008a,AMooetal,TreeSim1,TreeSim2,TStaMSte2012}.

For the purpose of this current work I need the Laplace 
transform of, $T$, the height of a Yule tree conditioned on $n$ tips at
present and the Laplace transform of the random variable $\tau$,
the time to the most recent common ancestor of two species randomly sampled out of $n$.
These have already been derived and studied in detail by \citet{KBarSSag2012}.
In addition to be able to incorporate the jump events I
will need to study the random variables $\Upsilon$ and $\upsilon$ --- 
the number of speciation events from the time of
origin of the tree until a randomly chosen tip species
and  the number of speciation events
that occurred on the lineage from the tree origin to the most recent common ancestor (excluding it)
of two randomly sampled tip species respectively.
\mbox{Figure \ref{tr}} visualizes these random variables. These two random variables
can also be seen as distances on the tree counted as number of edges,
$\Upsilon$ --- distance from the root of random tip, 
$\upsilon$ --- distance from the root of the most recent common ancestor of two 
randomly sampled tips.

\begin{figure}[!h]
\centering
\includegraphics[width=13cm]{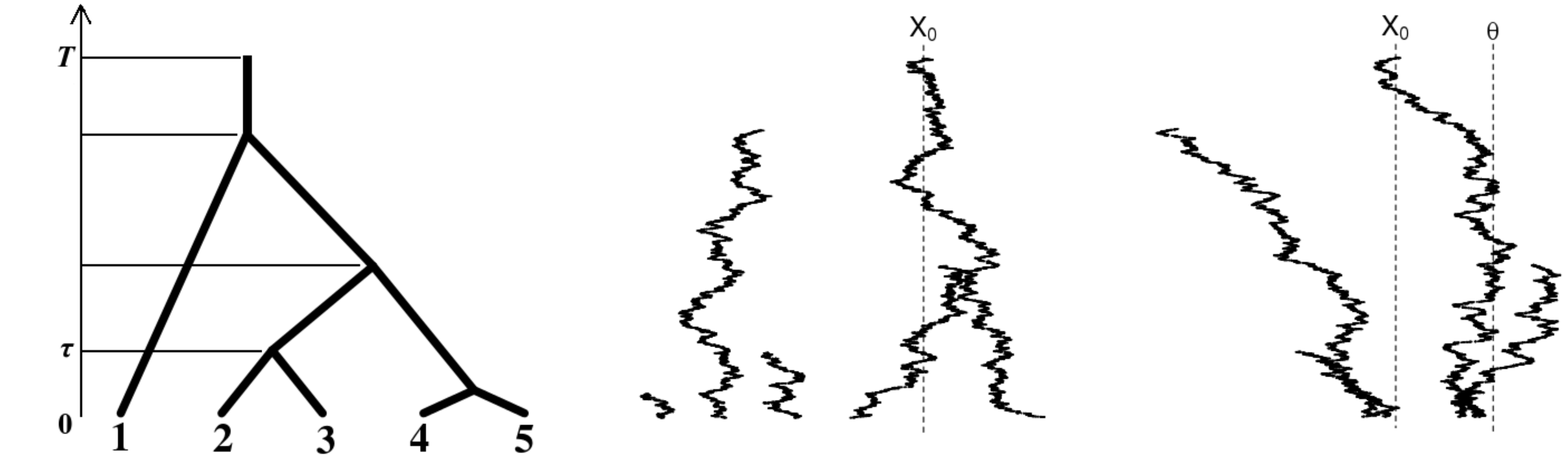}
\caption{
Left: a conditioned Yule tree ($\lambda=1$, $n=5$) with the different time components marked on it.
The height of the tree is $T$.
If species $2$ is ``randomly'' chosen then it would have $\Upsilon=3$ and if the pair of species
$2$, $3$ is ``randomly'' chosen 
then they would have $\upsilon=2$ and the time till their coalescent would be $\tau$. 
Center: a Brownian motion with jumps 
($X_{0}=0$, $\sigma_{a}^{2}=1$, $\sigma_{c}^{2}=5$) evolving on the tree
and right: an Ornstein--Uhlenbeck process with jumps ($\alpha=1$, $\theta=1$, $X_{0}=0$, $\sigma_{a}^{2}=1$, $\sigma_{c}^{2}=5$)
evolving on the tree.
Simulations were done using the TreeSim \citep{TreeSim1,TreeSim2} and mvSLOUCH \citep{KBaretal} R \citep{R} packages. 
}
\label{tr}
\end{figure}

I recall the following key lemma concerning a Yule tree,
\begin{lemma}[\citealp{KBarSSag2012,TreeSim1}]
In a Yule tree conditioned on $n$ contemporary species the 
probability that the coalescent of two randomly sampled tips occurred at
the $k$--th (counting forward in time) speciation event is,
\be \label{eqpikn} \pi_{k,n}={2(n+1)\over (n-1)(k+2)(k+1)},\ k=1,\ldots,n-1.\ee
\end{lemma}

I use the following notation,
\bd
\begin{array}{rcl}
H_{n} & = & \sum_{i=1}^{n} \frac{1}{n} \\
b_{n,y} & = & {1\over y+1}\cdot{2\over y+2}\cdots{n\over y+n} = \frac{\Gamma(n+1)\Gamma(y+1)}{\Gamma(n+y+1)},
\end{array}
\ed
where $\Gamma(\cdot)$ is the gamma function.
In the special case $y=j$, where $j$ is an integer, $b_{n,j}={1\over{n+j\choose j}}$.
Using this I recall \citep[from][]{KBarSSag2012} the Laplace transforms 
of $T$ and $\tau$,
\be
\E{e^{-y T}} = b_{n,y}
\label{eqLapT}
\ee 
and
\be
\E{e^{-y \tau}}  ={2-(n+1)(y+1)b_{n,y}\over(n-1)(y-1)}
\label{eqLaptau}
\ee
with $\E{e^{-\tau}} = {2\over n-1}(H_{n}-1)- {1\over n+1}$ as the limit of $y \rightarrow 1$.

I further recall \citep{KBarSSag2012,SSagKBar2012,MSteAMcK2001},
\be
\begin{array}{ccc}
\E{T} =  H_{n}, & ~~~ &
\E{\tau}  =  \frac{n+1}{n-1}H_{n} -\frac{2n}{n-1},
\end{array}
\ee
both behaving as $\ln n$ 
\citep[however see also][for situations with extinction present]{TGer2008a,TGer2008b,AMooetal,SSagKBar2012,TSta2008,TStaMSte2012}.

\subsection{Counting speciation events}
I defined the random variable $\Upsilon$ as the number of speciation events from the time of
origin of the tree until a randomly chosen tip species (see \mbox{Fig. \ref{tr}).} 
I can write $\Upsilon = \sum_{i=1}^{n-1} \mathbf{1}_{i}$,
where $\mathbf{1}_{i}$ is the indicator random variable of whether the 
$i$--th speciation event (counting from the first speciation event) is
on the randomly chosen lineage. The probability that $\mathbf{1}_{i}=1$
is $2/(i+1)$. The argument behind this is as follows, just before the $i$--th
speciation event there are $i+1$ points that need to coalesce. Exactly one of these
is on the lineage of interest. There are $\binom{i+1}{2}$ in total possibilities
of choosing the two points that will coalesce. Exactly $i$ will contain
the point of interest, so as I consider the Yule tree model
the probability is $i/\binom{i+1}{2}=2/(i+1)$. From this I get that 
\be\label{eqEUpsilon}
\E{\Upsilon}=2H_{n}-2 \sim 2\ln n. 
\ee

Now let us consider the situation that I sample two tip species
and am interested in the expectation of, $\upsilon$, the number of speciation events
that occurred on the lineage from the tree origin to their most recent common ancestor (excluding it), see \mbox{Fig. \ref{tr}.}
Using Eq. \eqref{eqpikn} 
and the following identity,
\be
\sum\limits_{k=1}^{n-1}\frac{H_{k}}{(k+1)(k+2)} = \frac{n-H_{n}}{n+1} \stackrel{n \rightarrow \infty}{\xrightarrow{\hspace*{1cm}}} 1
\ee
this equals,
\be \label{eqEupsilon}
\begin{array}{rcl}
\E{\upsilon} & = & \sum\limits_{k=1}^{n-1}\pi_{k,n}2(H_{k}-1) = 
\frac{4(n+1)}{n-1}\sum\limits_{k=1}^{n-1} \sum\limits_{i=1}^{k} \frac{1}{(k+2)(k+1)} \frac{1}{i} -2
\\&=& 
\frac{4}{n-1}\left(n-H_{n}\right)-2 \stackrel{n \rightarrow \infty}{\xrightarrow{\hspace*{1cm}}} 2.
\end{array}
\ee

The same formulae for $\E{\Upsilon}$ and $\E{\upsilon}$ were derived by \citet{MSteAMcK2001}
alongside the distribution and variance of $\Upsilon$.

I can add a variation to the above discussion by marking each speciation event with probability
$p$. Let then $\Upsilon^{\ast} \le \Upsilon$ count the number of marked events
along a randomly chosen lineage. Obviously 
\be
\E{\Upsilon^{\ast}} =  p\E{\Upsilon}=2p(H_{n}-1).
\ee
Similarly let $\upsilon^{\ast} \le \upsilon$ count the number of marked speciation events
from the tree origin to the most recent common ancestor of two randomly chosen tip species.
Again 
\be
\E{\upsilon^{\ast}} = p\E{\upsilon} = 2p\left(\frac{2}{n-1}\left(n-H_{n}\right)-1\right).
\ee

\begin{figure}
\begin{center}
\includegraphics[width=0.4\textwidth]{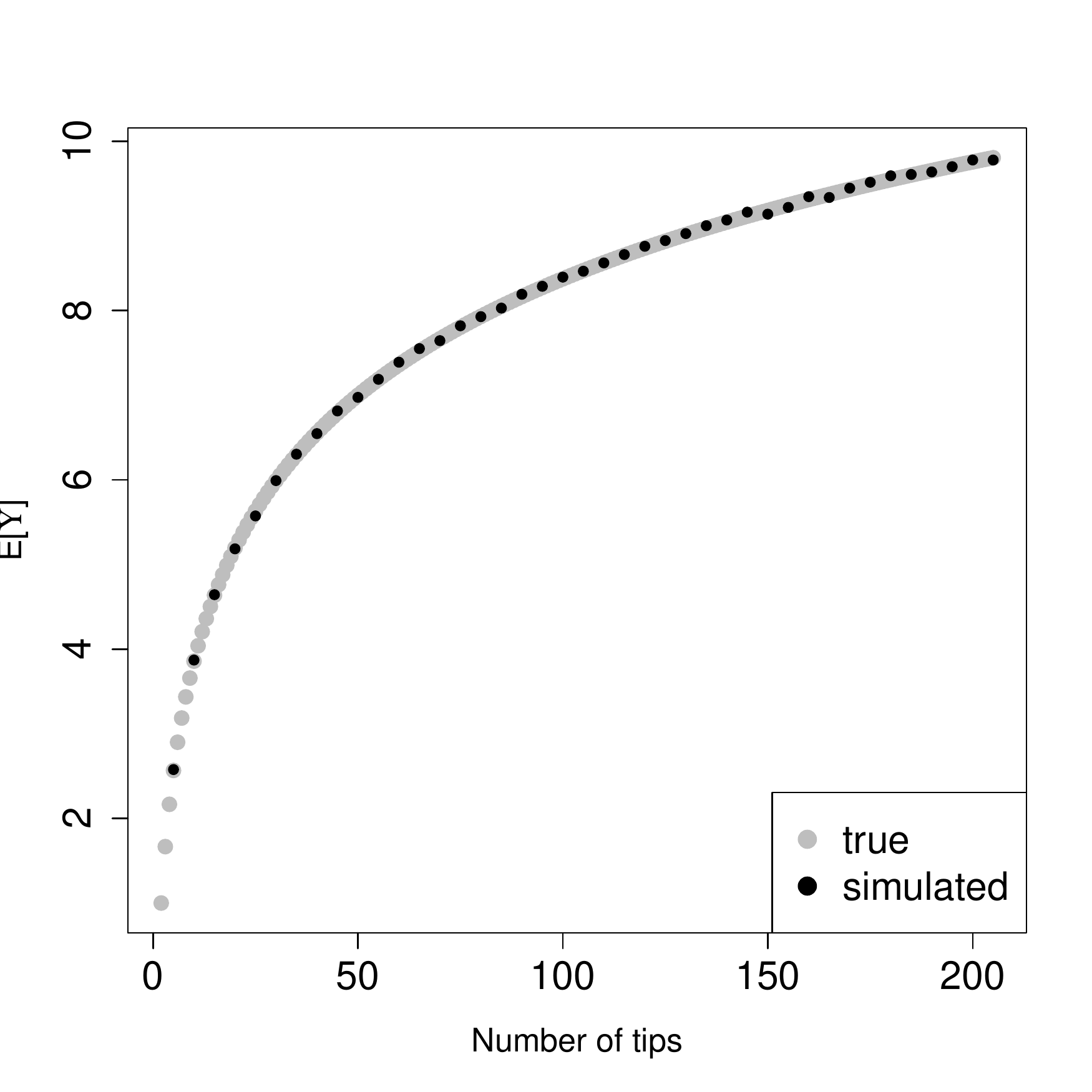}
\includegraphics[width=0.4\textwidth]{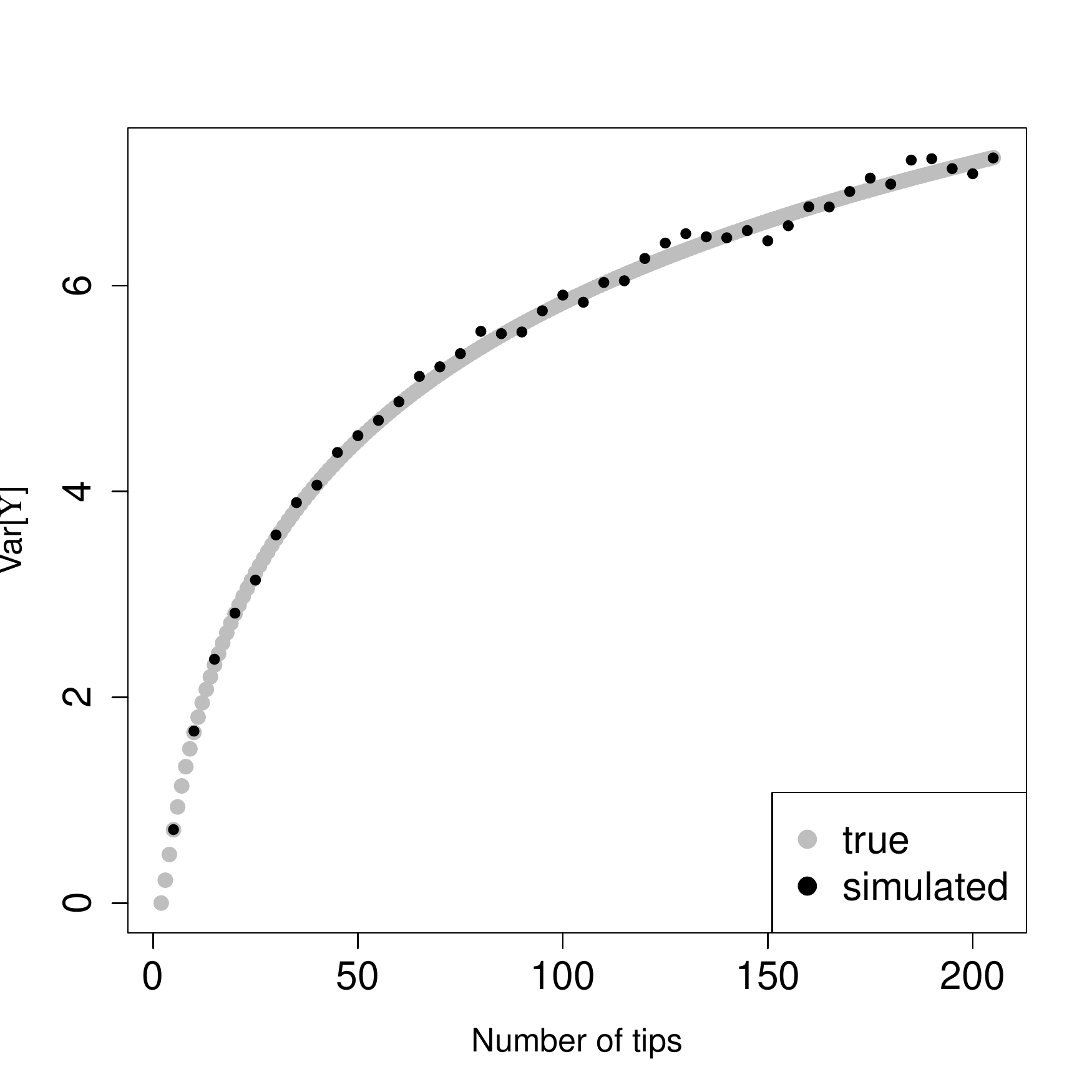} \\
\includegraphics[width=0.4\textwidth]{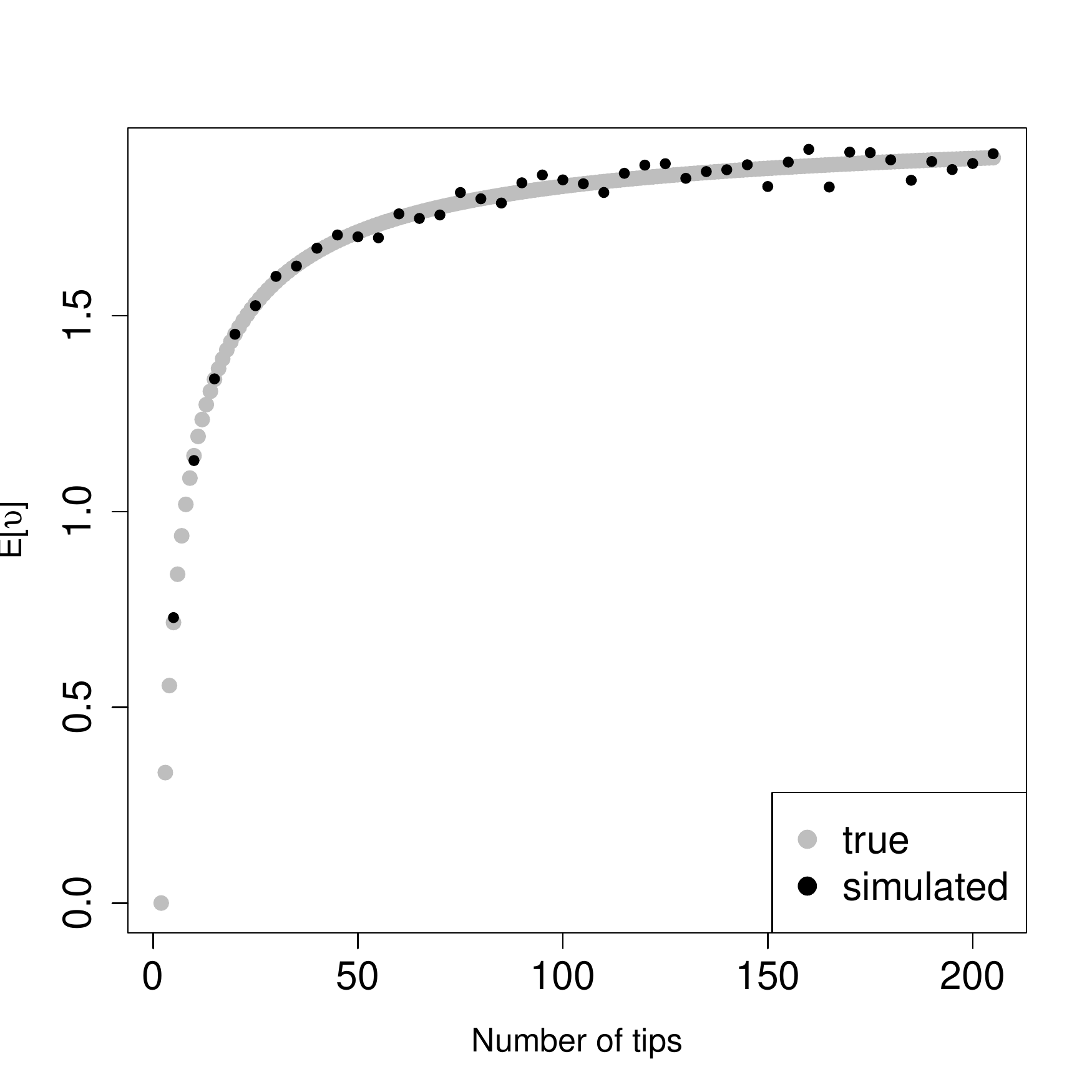}
\includegraphics[width=0.4\textwidth]{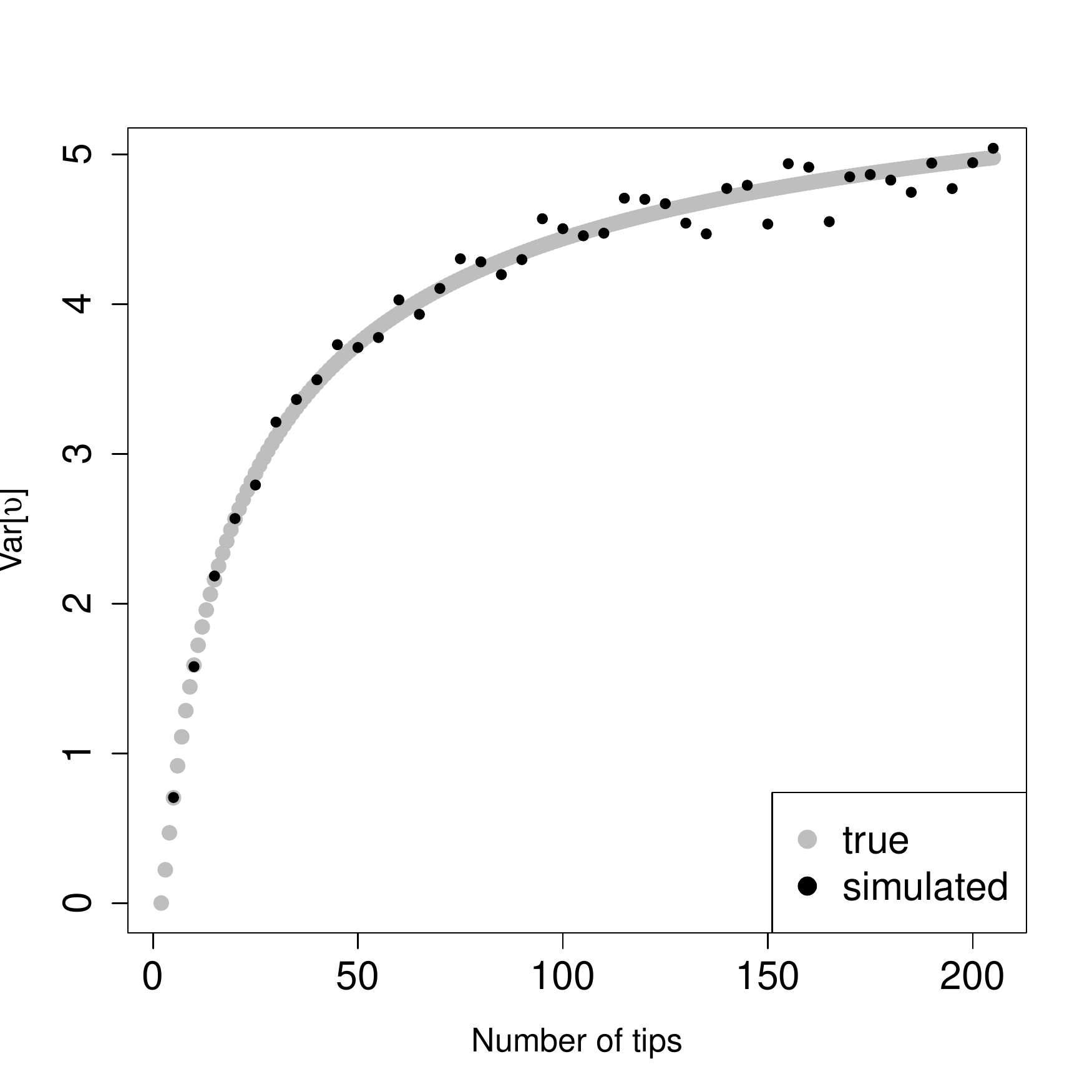}
\end{center}
\caption{Simulated and true values of 
$\E{\Upsilon}$, $\Var{\Upsilon}$ (top) and $\E{\upsilon}$, $\Var{\upsilon}$ (bottom). 
Each point comes from $10000$ simulated Yule trees.
$\Var{\upsilon}$ grows very slowly and with $n=200$
it is still rather distant from its asymptotic value of $6$.
}
\end{figure}

I further derive here the probability generating functions of $\Upsilon$ and $\upsilon$. 
They do not depend on the speciation rate $\lambda$ (a scale parameter for branch lengths 
when conditioning on $n$) as $\Upsilon$ and $\upsilon$ 
are topological properties.
Notice in the theorem below that 
$b_{n-1,2s-1}$ is well defined because $2s-1>-1$ (as $2s>0$)
and therefore consequently the probability generating functions are well defined for all $s>0$. 
One should also not forget that in the case of there being only two extant species
$\upsilon$ equals $0$ due to there being only one speciation event --- the most recent common ancestor of the only pair 
of species --- and this is not counted by $\upsilon$.
\begin{theorem}\label{thPGF}
The probability generating function for $\Upsilon$ is,
\be
\E{s^{\Upsilon}} = \frac{1}{n}\frac{1}{b_{n-1,2s-1}}, s>0,
\ee
while for $\upsilon$,
\be
\E{s^{\upsilon}} = \left\{
\begin{array}{ll}
\frac{1}{(n-1)(3-2s)}\left((n+1)-\frac{2}{nb_{n-1,2s-1}}\right) & 0<s \neq \frac{3}{2}, n\ge 2 \\
\frac{2(n+1)}{n-1}\left(H_{n+1}-\frac{5}{3} \right) & s=\frac{3}{2}, n\ge 2.
\end{array}
\right.
\ee
\end{theorem}
Using the value of the second derivatives of the probability generating function at $1$ I can calculate the variance of
$\Upsilon$ \citep[obtained earlier in][in a different manner]{MSteAMcK2001} and $\upsilon$ as,
\be
\begin{array}{rclcl}
\var{\Upsilon} & = & 2\left(H_{n}-1-2(H_{n,2}-1) \right) & \sim & 2\ln n, \\
\var{\upsilon} & = & 6-8\frac{n+1}{(n-1)^{2}}(H_{n}-1)^{2}-4\frac{H_{n}-1}{n-1}+8\frac{H_{n,2}-1}{n-1} & \rightarrow & 6
\end{array}
\ee
where $H_{n,2}=\sum\limits_{k=1}^{n}\frac{1}{k^{2}}$. 

\subsubsection{Connection to total tree area}
The \emph{total area} of a tree $T$ with $n$ tips, is defined \citep[as in][]{AMirFRos2010},
\be
D_{n}(T) = \sum\limits_{1\le i < j \le n}d_{T}(i,j),
\ee
where $d_{T}(i,j)$ is the distance between nodes $i$ and $j$. This distance
can be counted in two ways, either as the number of edges on the path between the two nodes
or as the number of vertices. The latter will be one less than the former. 
For a Yule tree I can calculate the expectation
of the total area as 
\begin{itemize}
\item $\binom{n}{2}(2\E{\Upsilon}-2\E{\upsilon})$ with the first definition of $d_{T}(i,j)$,
\item $\binom{n}{2}(2\E{\Upsilon}-2\E{\upsilon}-1)$ with the second definition of $d_{T}(i,j)$.
\end{itemize}
Plugging in the values for $\E{\Upsilon}$ and $\E{\upsilon}$, Eqs. \eqref{eqEUpsilon} and \eqref{eqEupsilon}, I get,
\be
\E{D_{n}} = 2n(n+1)H_{n}-4n^{2} \sim 2n^{2}\ln n
\ee
and
\be
\E{D_{n}} = 2n(n+1)H_{n}-4n^{2}-\frac{n(n-1)}{2} \sim 2n^{2}\ln n
\ee
respectively, depending on the definition of $d_{T}(i,j)$. These results are 
the same as in the literature \citep[][in the case of the first definition]{AMirFRosLRot2013} 
and \citep[][in the case of the second definition]{WMul2011}.
\citet{AMirFRosLRot2013} claim that the expectation of the total area of a Yule tree calculated by \citet{WMul2011}
contains an error, however the discrepancy between the results of these
two studies comes from them using the different definitions of the distance between two tips.

\section{Models with punctuated evolution}
\label{secPEq}
Stochastic models for continuous trait evolution are commonly based on a stochastic 
differential equation (SDE) of the type,
\be\label{eqSDEdiff}
\ud X(t) = \mu(t,X(t))\ud t + \sigma_{a}\ud B_{t},
\ee
along the phylogenetic tree \citep[see e.g.][]{KBaretal,MButAKinOUCH,JFel85,THan1997,THanJPieSOrzSLOUCH,ALabJPieTHan2009}. 
At speciation times
this process divides into two processes evolving independently from that point.
Very often it is assumed that the natural logarithm of the trait evolves according to this
SDE \citep[for some motivation for this see e.g.][]{KBaretal,FBok2002,JHux32,MSav79}.

It is straightforward to include in this framework a mechanism for modelling 
cladogenetic evolution. I consider two possible mechanisms.
The first  one is that just after each speciation point in each daughter lineage
with probability $p$ a jump (mean $0$, variance $\sigma_{c}^{2} < \infty$) 
takes place. 
The second mechanism is 
that at each speciation point one adds to the phenotype process of a
randomly (with probability $0.5$) chosen daughter lineage a 
mean $0$, variance $\sigma_{c}^{2} < \infty$ jump. The other daughter lineage is 
then not affected by the jump.
This can be interpreted
as some change in the newborn species that drove the species apart.
At the present second--moment level of analysis this model is equivalent to the first one
with $p=0.5$. One cannot distinguish between
these two punctuated change mechanisms unless one actually observes the speciation and jump
patterns. A jump at each daughter lineage with probability $0.5$ and a jump in exactly
one randomly chosen daughter lineage will have the same effect on the
first and second sample moments. This is because they only depend on the expectation and
variance of the jump and number of jumps in the common part of the two lineages.
In the simulations shown in this work I assume normality of the jump
but all the results will hold for any mean $0$, finite variance jump.

I will study two currently standard evolutionary models the Brownian motion  
and Ornstein--Uhlenbeck process both expanded to include a punctuated equilibrium component. 
In line with previous work \citep{KBarSSag2012,SSagKBar2012} I do not condition
on a given phylogenetic tree but assume a branching process (here conditioned
Yule tree) for the phylogenetic tree. In such a case, a number of relevant model
properties \citep[see][]{KBarSSag2012,SSagKBar2012} can be conveniently described in terms of the variance of 
the trait of a randomly sampled tip species and the covariance (or correlation) between two randomly sampled
tip species. This correlation \citep[not to be confused with correlation between
traits which can result from developmental constraints][]{JChe1984} called the interspecies correlation coefficient
\citep{SSagKBar2012} is a consequence of phylogenetic inertia. 
Extant species will be correlated due to their shared ancestry. How strong that correlation
is will depend on how much time has passed since their
divergence (stochastic tree model) and mode of evolution (stochastic trait model).

The interspecies correlation is similar to what \citet{JCheMDowWLeu1985}
describe as the phylogenetic autocorrelation. \citet{JCheMDowWLeu1985}
repeated after \citet{RRie1978} and \citet{CWad1957} that the similarity
of a trait between species can be indicative of its functional role.
A highly central trait should be significantly correlated while 
a peripheral one should have more freedom to vary. This is in
line with thinking of all of the traits interacting with each other.
If a trait is central then all the other traits will have adapted
to working with it and any change in it would cause 
all the other traits to be away from their optimum
and hence potentially cause a significantly larger misadaptation
than change in a trait that is not so central.

Below I discuss the correlation function (with detailed derivations in Appendix A). 
It will turn out that in the case of punctuated evolution 
a relevant parameter is \mbox{
$\kappa=\frac{2p\sigma_{c}^{2}}{\sigma_{a}^{2}/\lambda+2p\sigma_{c}^{2}}$.}
Knowing whether the gradual or jump component dominates tells us whether the correlation is
due to shared branch length (gradual dominates) or shared number of speciation events (jumps dominate).
One can also recognize a similarity to measurement error theory. If one thinks of the jumps as
``errors'' added to the trajectory of the evolving diffusion process then $1-\kappa$ 
could be thought of as the reliability ratio or  measurement error correction factor \citep[see e.g.][]{JBuo,WFul1987,THanKBar2012}.

The work here is also another step in introducing L\'evy processes
to the field of phylogenetic comparative methods \citep[see also][]{KBar2012,MLanJSchMLia2013}.
The framework of L\'evy processes is very appealing as it naturally includes punctuated change,
i.e. jumps in the evolution of continuous traits.

I introduce the following notation, by $X(t)$ I will mean the phenotype process,
$X$ will denote the trait value of randomly sampled tip species, while
$X_{1}$ and $X_{2}$ will denote the trait values of a randomly sampled
pair of tip species.

\subsection{Unconstrained evolutionary model}
The Brownian motion model \citep{JFel85} can be described by the following SDE,
\be \ud X(t) = \sigma_{a} \ud B(t), ~~~ X(0) = X_{0}. \ee
On top of this at each speciation point,  
each daughter lineage 
has with probability $p$ a mean $0$, variance $\sigma_{c}^{2} < \infty$ jump added to it. 
As the jump is mean zero it does not change the expectation of the value of a tip species.
Below I consider the variance of a randomly sampled tip species and the covariance
and correlation between two randomly sampled tip species.

To derive the following theorem I rely on previous results \citep{KBarSSag2012,SSagKBar2012}.
\begin{theorem}
The interspecies correlation coefficient for a phenotype evolving as a Brownian motion with jumps 
on top of a conditioned Yule tree with speciation rate $\lambda=1$ is,
\be
\rho_{n} = 
\frac{\frac{2(n-H_{n})}{n-1}-\kappa}{H_{n}-\kappa}.
\label{eqCorYuleBM}
\ee
\end{theorem}
\begin{proof}
Due to the jumps being independent of the evolving Brownian motion process
\be
\Var{X} = \sigma_{a}^{2}\E{T} + p\sigma^{2}_{c}\E{\Upsilon} = (\sigma_{a}^{2} + 2p\sigma^{2}_{c})H_{n}-2p\sigma_{c}^{2}
\label{eqVarYuleBM}
\ee
and the covariance between two randomly sampled tip species,
\be
\cov{X_{1}}{X_{2}} = \sigma_{a}^{2}\E{T-\tau} + p\sigma_{c}^{2}\E{\upsilon} =
\frac{2(n-H_{n})}{n-1}(\sigma_{a}^{2}+2p\sigma_{c}^{2})-2p\sigma_{c}^{2}.
\label{eqCovYuleBM}
\ee
Taking the quotient of these two values results in the desired formula for $\rho_{n}$.
\end{proof}
Comparing with the correlation coefficient for the Brownian motion process calculated by
\citet{SSagKBar2012} adding jumps causes both the numerator and denominator to be corrected by $\kappa$.
The following asymptotic behaviour can be directly seen,
\be
\begin{array}{rcl}
\var{X} & \sim & (\sigma_{a}^{2}+2p\sigma_{c}^{2})\ln n, \\
\cov{X_{1}}{X_{2}} & \sim & 2\sigma_{a}^{2}+2p\sigma_{c}^{2}+O(\ln n/n), \\
\rho_{n} & \sim & \frac{2-\kappa}{\ln n}. 
\end{array}
\ee 

Taking the derivative of the correlation in terms of $\kappa$ I find that it is negative and so:
\be
\rho_{n} \stackrel{\kappa \rightarrow 1}{\xrightarrow{\hspace*{1cm}}}  \frac{n-2(H_{n}-0.5)}{(n-1)(H_{n}-1)} \ge 0~\mathrm{for}~n\ge 2,
\ee
which is a monotonically decreasing convergence.

In the above I assumed that the speciation rate is $\lambda=1$. 
This restriction does not change the validity
of the results as changing $\lambda$ is equivalent to rescaling the branch lengths by its inverse. 
As mentioned because I have conditioned on $n$ this 
has no effect on the topology, and therefore does not effect $\Upsilon$ and \mbox{
$\upsilon$.}  Consequently
a Yule--Brownian--motion--jumps model for the extant species trait sample with parameters $(\sigma_{a}^{2}, \sigma_{c}^{2}, X_{0},\lambda)$
is equivalent to one with parameters $(\sigma_{a}^{2}/\lambda, \sigma_{c}^{2},  X_{0},1)$.

\begin{figure}
\centering
\includegraphics[width=0.32\textwidth]{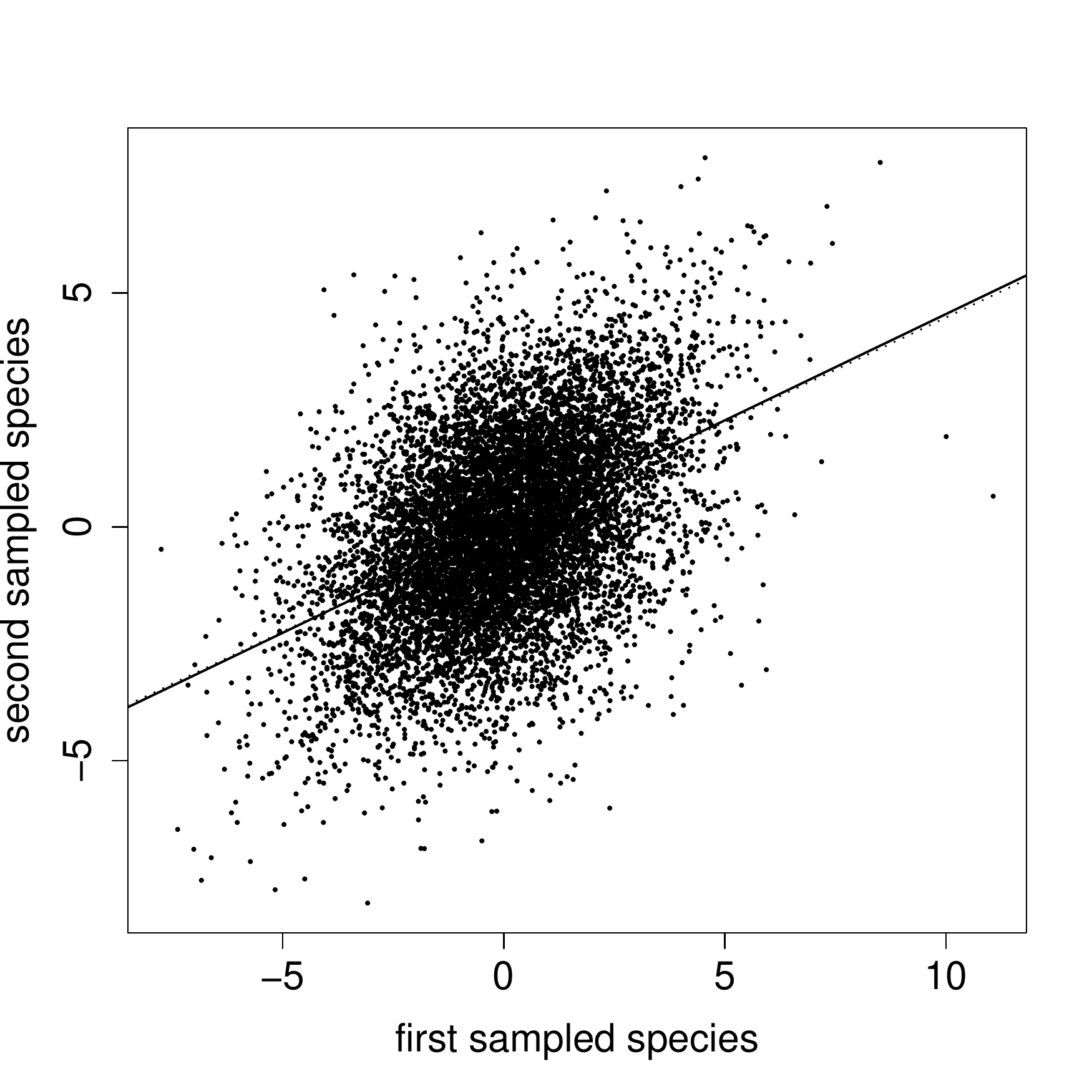}
\includegraphics[width=0.32\textwidth]{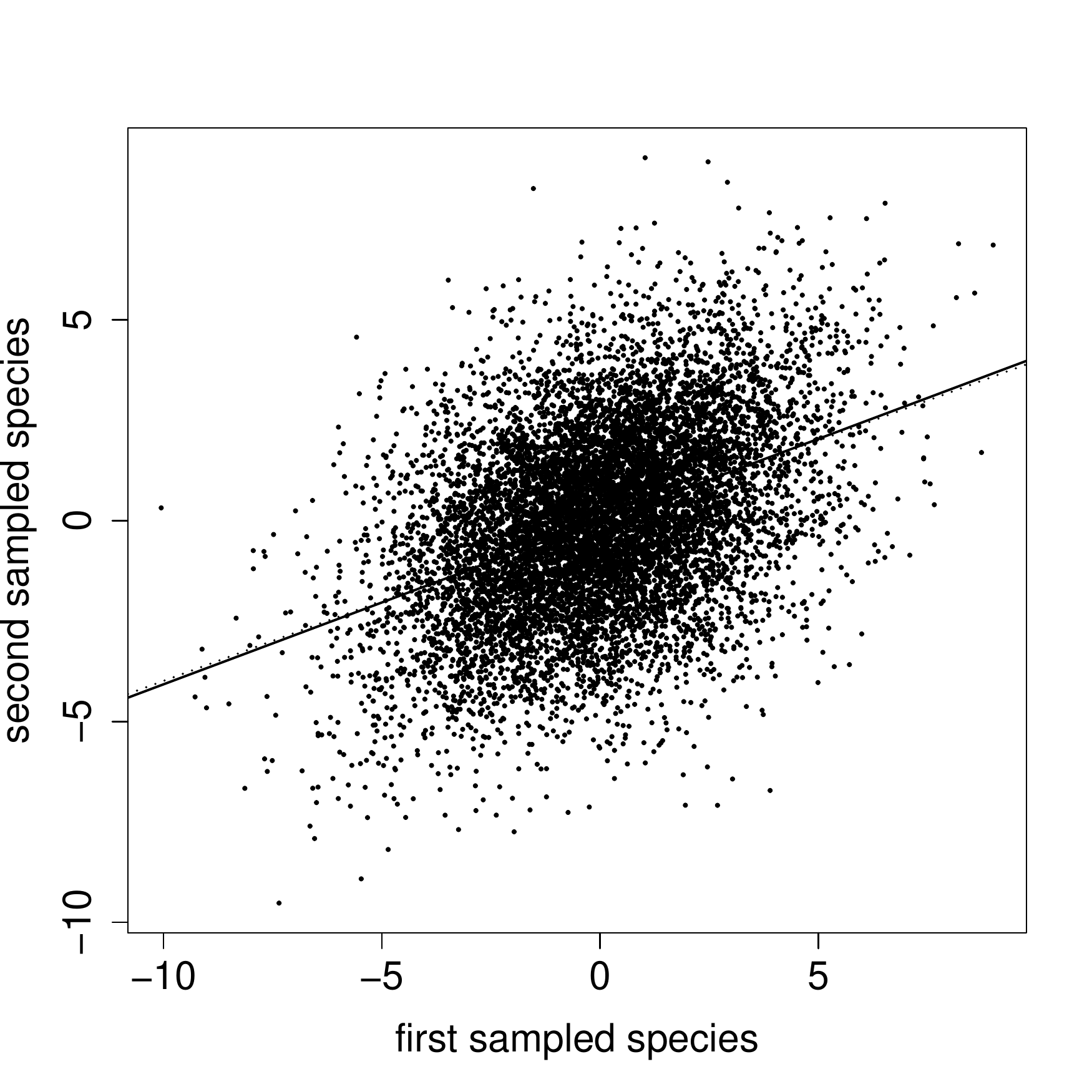}
\includegraphics[width=0.32\textwidth]{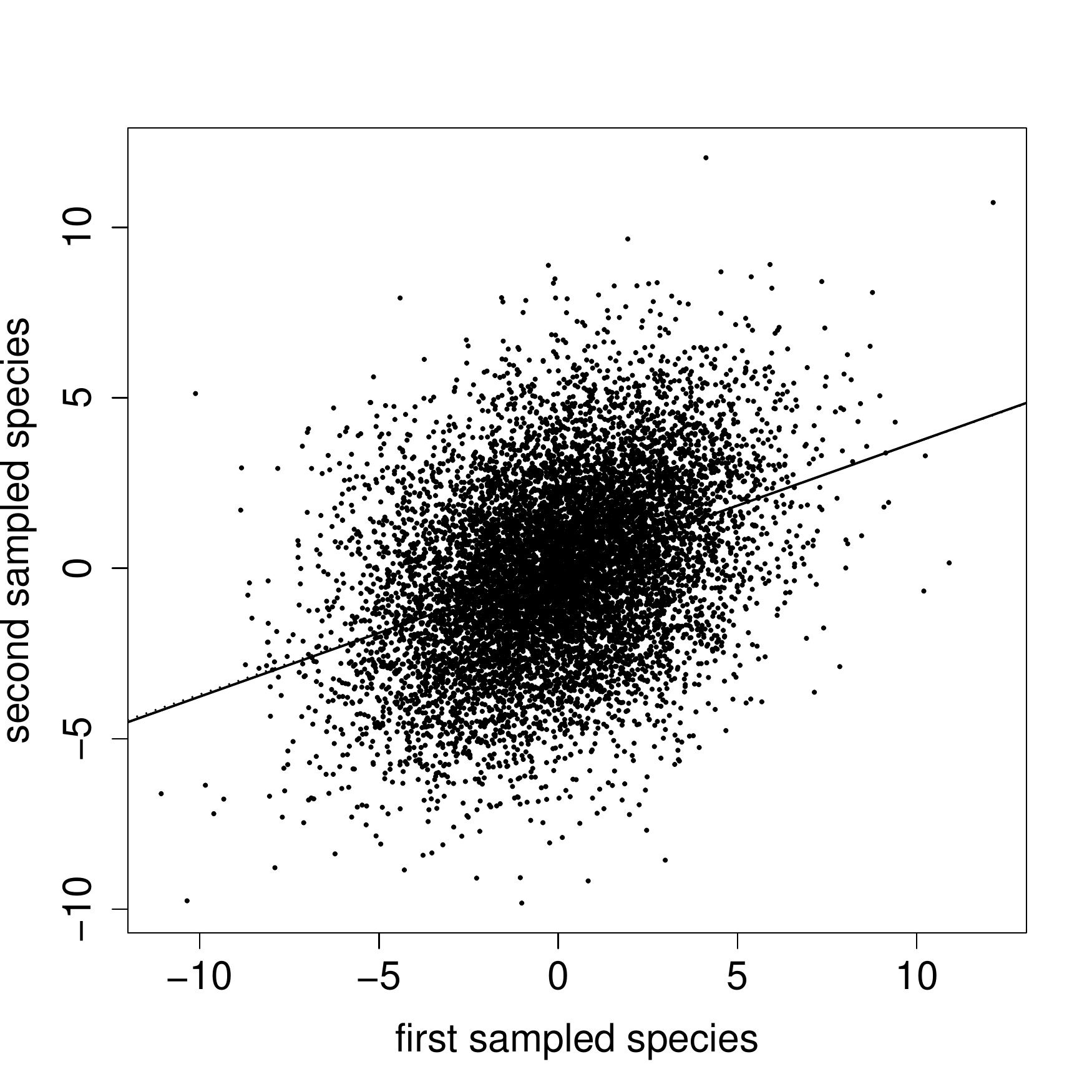} \\
\includegraphics[width=0.32\textwidth]{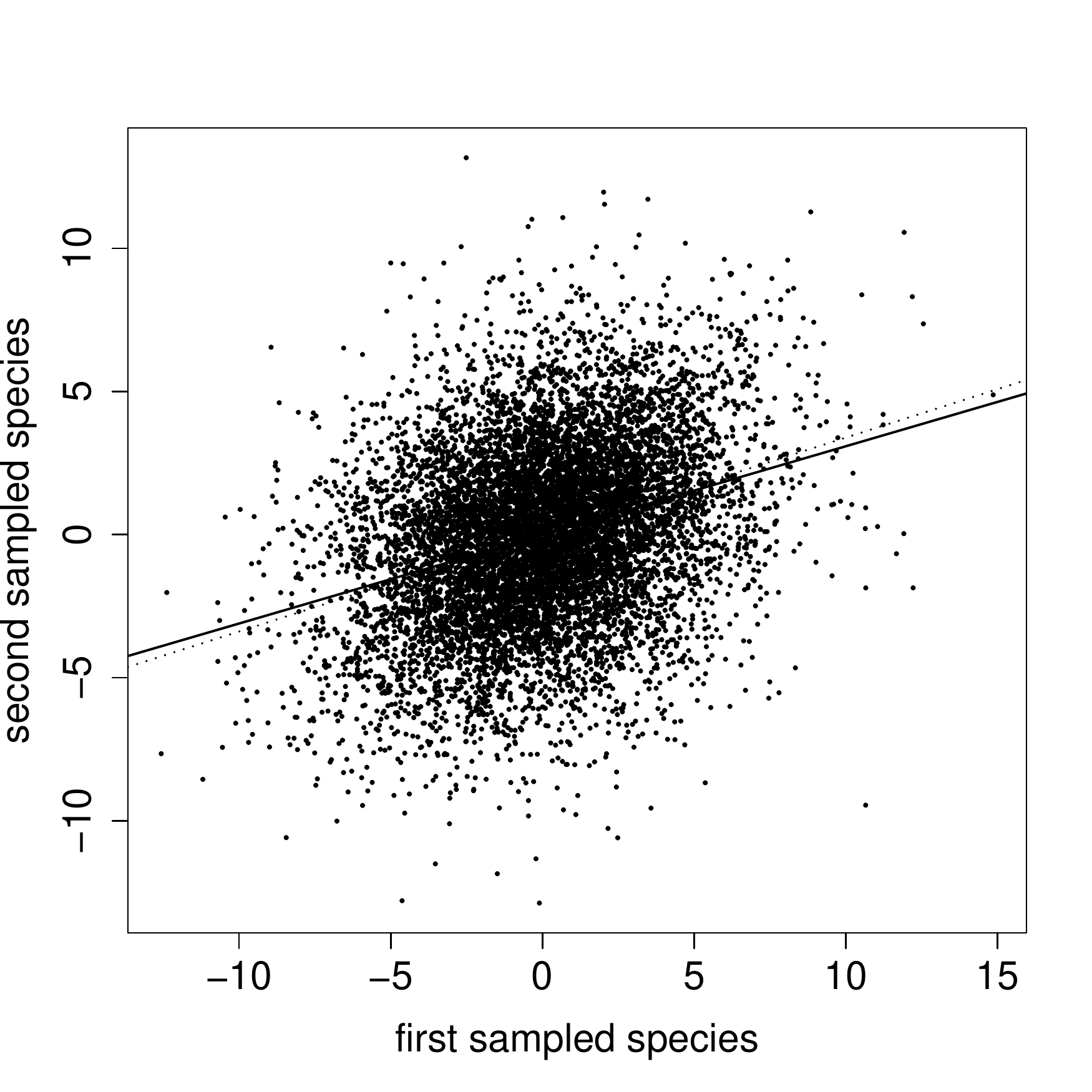} 
\includegraphics[width=0.32\textwidth]{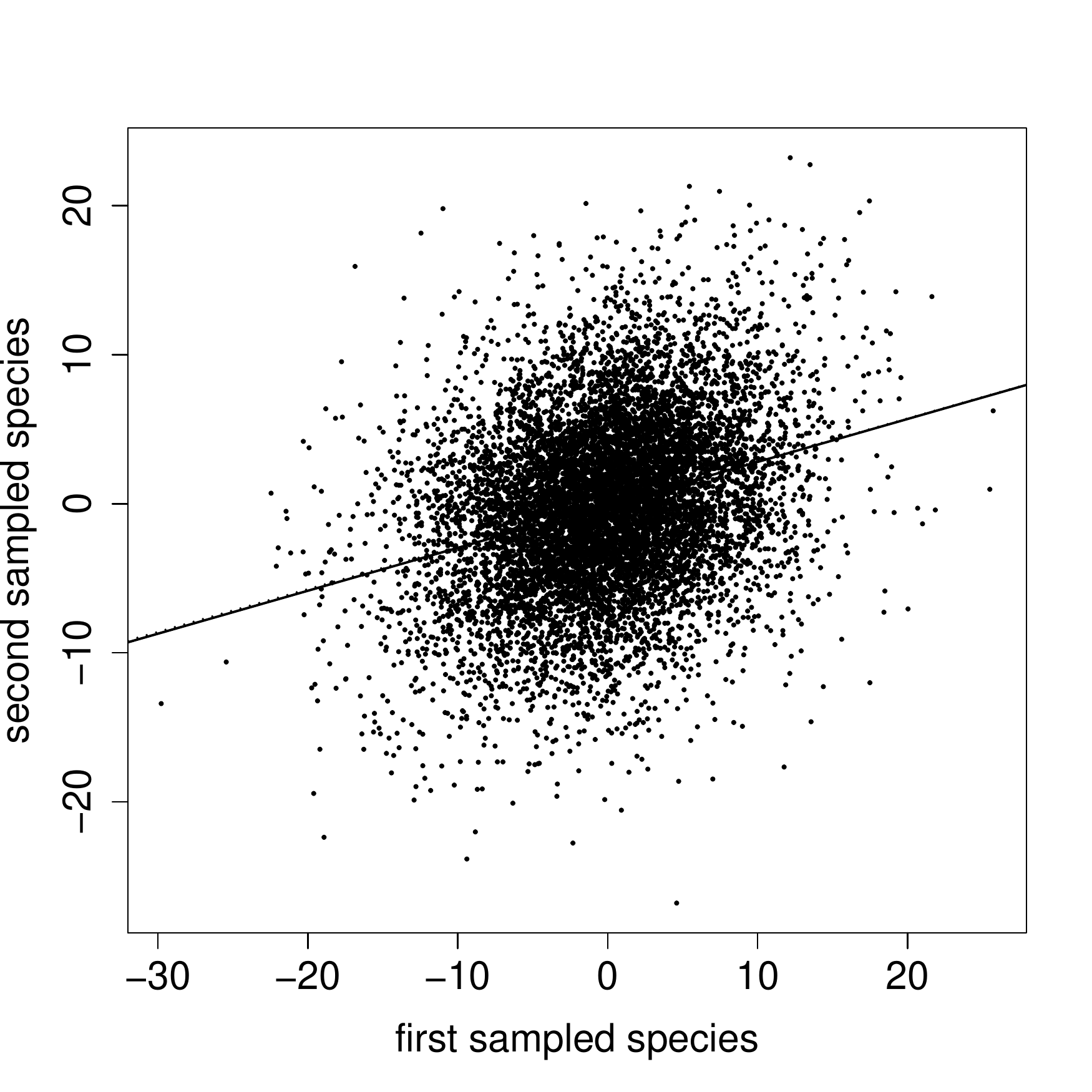}
\includegraphics[width=0.32\textwidth]{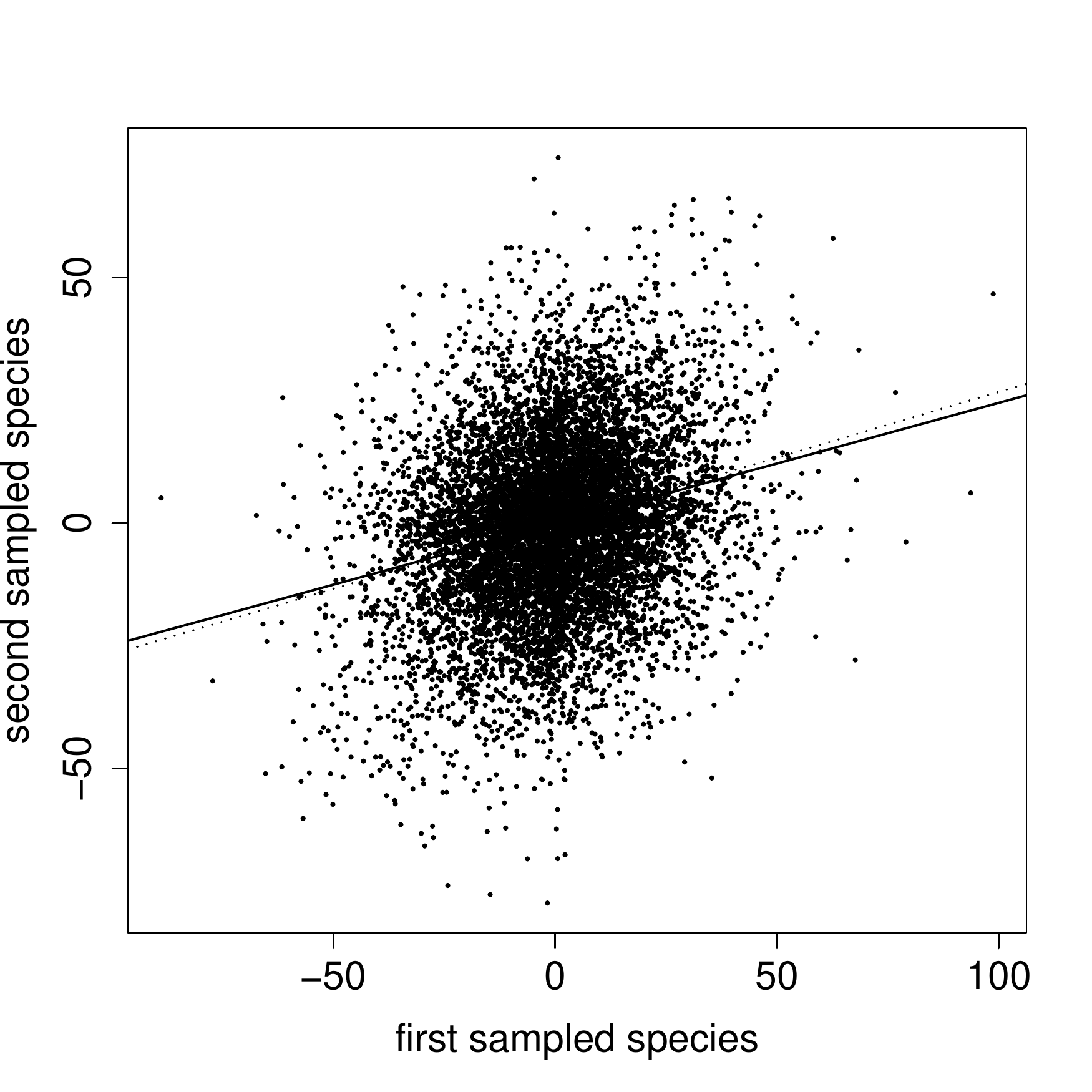}
\caption{
Regression lines fitted to simulated data (thick line) mostly indistinguishable from the true regression line
$y=\rho_{n}x$ (dotted line) with $\rho_{n}$ given by the exact formula (Eq. \ref{eqCorYuleBM}) for different
values of $\kappa$ in the Yule--Brownian--motion--jumps model. Top row from left to right :
$\kappa=0.0099, 0.3333, 0.5$, bottom row from left to right : $\kappa = 0.6667,0.9091,0.9901$.
In all cases the jump is normally distributed with mean $0$ and variance $\sigma_{c}^{2}$, $p=0.5, X_{0}=0$, $\sigma_{a}^{2}=1$
and $n=30$.}
\end{figure}

\begin{figure}
\begin{center}
\includegraphics[width=0.5\textwidth]{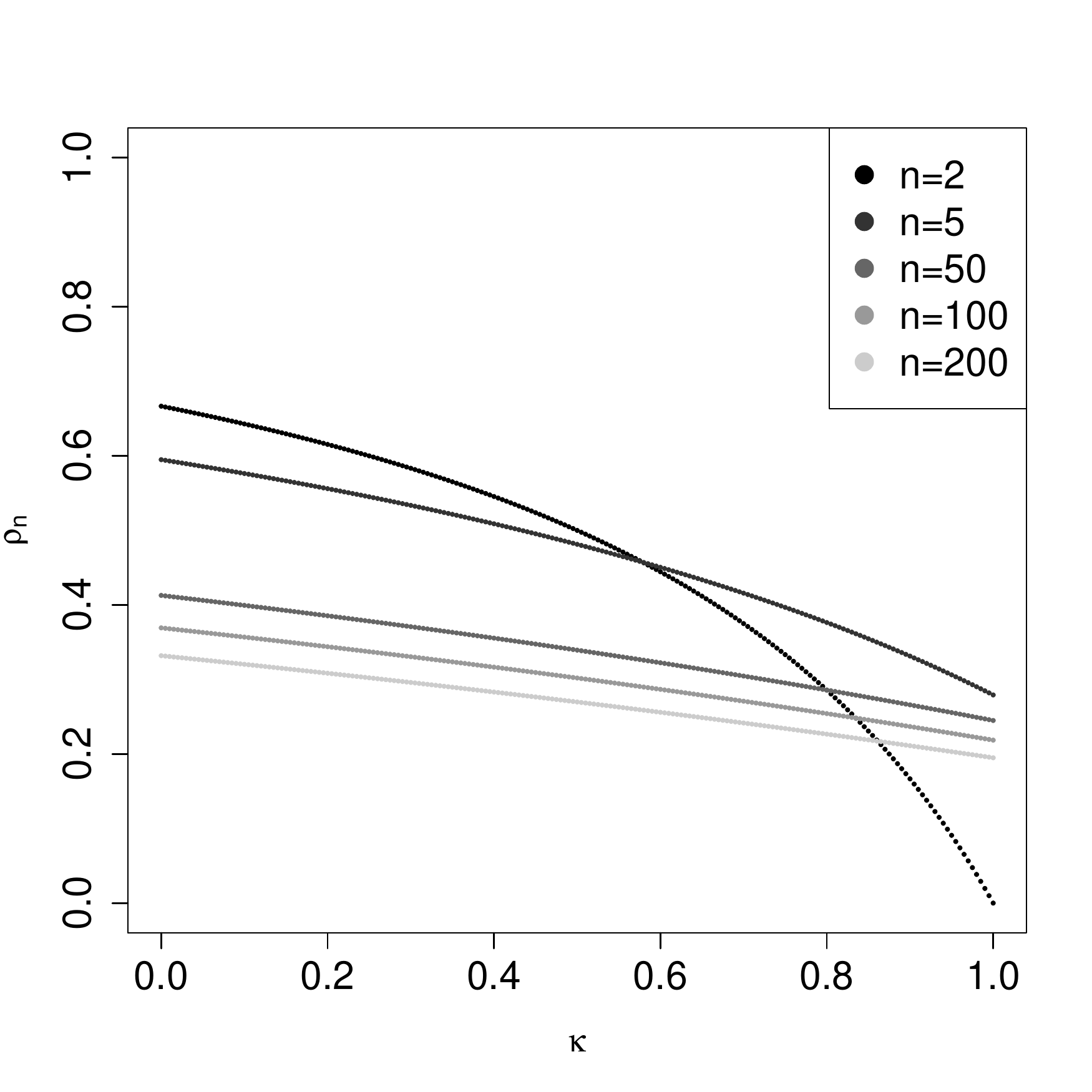}
\caption{
Interspecies correlation coefficient for the Yule--Brownian--motion--jumps model for different values of $n$.
}
\end{center}
\end{figure}

\subsection{Constrained model}
The basic stochastic process used to model adaptation in the phylogenetic comparative methods
field is the Ornstein--Uhlenbeck process \citep{MButAKinOUCH,JFel1988,THan1997,THanJPieSOrzSLOUCH,ALabJPieTHan2009},
\be
\ud X(t) = -\alpha(X(t)-\theta) \ud t + \sigma_{a}\ud B(t), ~~~ X(0)=X_{0}
\ee
and to it I introduce the jump component in the same fashion. 
At each speciation point a randomly chosen daughter lineage
has a mean $0$, variance \mbox{
$\sigma_{c}^{2} < \infty$} jump added to it. 
As discussed in the introduction this model is appealing as it allows
one to combine gradual and punctuated evolution in one framework.

However, in this model a jump could also affect the mean value and so a more careful 
treatment is necessary, see the proof in Appendix A. Again I can assume $\lambda=1$ as the following 
two parameter sets are equivalent,
$(\alpha,\sigma_{a}^{2}, \sigma_{c}^{2}, \theta, X_{0},\lambda)$ and
$(\alpha/\lambda, \sigma_{a}^{2}/\lambda, \sigma_{c}^{2}, \theta,  X_{0},1)$ 
for a Yule--Ornstein--Uhlenbeck--jump model.

To describe the Yule--Ornstein--Uhlenbeck--jump model I introduce the following 
convenient parameter,
$\delta=\vert X_{0}-\theta \vert/(\sqrt{\sigma_{a}^{2}/2\alpha})$
(distance of the starting value from the optimal one scaled by the stationary standard deviation of the Ornstein--Uhlenbeck process).
\begin{theorem}\label{thRhon}
The mean, variance, covariance, correlation values in a Yule--Ornstein--Uhlenbeck--jump model with $\lambda=1$ are,
\be
\begin{array}{rcl}
\E{X} & = & b_{n,\alpha}X_{0} + (1-b_{n,\alpha})\theta, \\
\Var{X} & = & \frac{\sigma_{a}^{2} + 2p\sigma_{c}^{2}}{2\alpha} \left((1-\kappa)V^{(n)}_{a}(\alpha,\delta) + \kappa V^{(n)}_{c}(\alpha)\right), \\
\cov{X_{1}}{X_{2}} & = & \frac{\sigma_{a}^{2} + 2p\sigma_{c}^{2}}{2\alpha} \left((1-\kappa)C^{(n)}_{a}(\alpha,\delta) + \kappa C^{(n)}_{c}(\alpha)\right), \\
\rho_{n} & = &\frac{(1-\kappa)C^{(n)}_{a}(\alpha,\delta)+\kappa C^{(n)}_{c}(\alpha)}{(1-\kappa)V^{(n)}_{a}(\alpha,\delta)+\kappa V^{(n)}_{c}(\alpha)},
\end{array}
\ee
where,
\bd
\begin{array}{lll}
C^{(n)}_{a}(\alpha,\delta) & = & 
\left\{
\begin{array}{ll}
\frac{2-(n+1)(2\alpha+1)b_{n,2\alpha}}{(n-1)(2\alpha-1)} - b_{n,2\alpha} 
 +\delta^{2}\left(b_{n,2\alpha} - b_{n,\alpha}^{2} \right), & 0<\alpha \neq 0.5, \\
\frac{2}{n-1}\left(H_{n}-1 \right) - \frac{2}{n+1} 
+\delta^{2}\left(\frac{1}{n+1} - b_{n,0.5}^{2} \right), & \alpha=0.5,
\end{array}
\right. \\
C^{(n)}_{c}(\alpha) & = & 
\left\{
\begin{array}{ll}
\frac{2-(2\alpha n -2\alpha +2)(2\alpha+1)b_{n,2\alpha}}{(n-1)(2\alpha-1)}, & 0<\alpha \neq 0.5, \\
\frac{2}{n-1}\left(H_{n}  - \frac{5n-1}{2(n+1)} \right), & \alpha=0.5,
\end{array}
\right.\\
V^{(n)}_{a}(\alpha,\delta) & = & 1-b_{n,2\alpha}+\delta^{2}(b_{n,2\alpha}-b_{n,\alpha}^{2}), \\
V^{(n)}_{c}(\alpha)& = &1-(1+2\alpha)b_{n,2\alpha}.
\end{array}
\ed
\end{theorem}
The exact final formula in Theorem \ref{thRhon} depends on whether $\alpha=0.5$ or \mbox{
$\alpha\neq 0.5$} \citep[see][for a discussion on this]{KBarSSag2012}.
$\var{X}$ and $\cov{X_{1}}{X_{2}}$ are made up of two distinct components: one from the Ornstein--Uhlenbeck
anagenetic evolution and the other from the cladogenetic ``jump'' evolution. 

Using that as $n\rightarrow \infty$, $b_{n,y}\sim \Gamma(1+y)n^{-y}$
due to the behaviour of the beta function $B(n,\alpha) \sim \Gamma(\alpha)n^{-\alpha}$
I obtain the following asymptotic behaviour for the variance,
\be
\begin{array}{rcl}
\var{X} & = & \frac{\sigma_{a}^{2}+2p\sigma_{c}^{2}}{2\alpha} + O(n^{-2\alpha}).
\end{array}
\ee
As $\alpha \rightarrow 0$ by expanding the $b_{n,y}$ symbol I arrive at the limit,
\bd
\begin{array}{rcl}
\frac{1-b_{n,2\alpha}}{2\alpha} & \stackrel{\alpha \rightarrow 0}{\xrightarrow{\hspace*{1cm}}} & H_{n} 
\end{array}
\ed
and by this the variance converges on that of a Yule--Brownian--motion--jumps model as $\alpha\rightarrow 0$.
Additionally using 
the de L'H\^ospital rule 
I obtain that the covariance converges on that of a Yule--Brownian--motion--jumps model
as $\alpha\rightarrow 0$ and as $n\rightarrow \infty$ 
the covariance behaves as,
\bd
\cov{X_{1}}{X_{2}} \sim
\frac{\sigma_{a}^{2}+2p\sigma_{c}^{2}}{2\alpha} \cdot
\left\{
\begin{array}{ll}
\left((1-\kappa)C_{a}(\alpha,\delta)+ \kappa C_{c}(\alpha) \right)n^{-2\alpha},  & 0< \alpha<0.5,\\
2 n^{-1}\ln n, & \alpha=0.5,\\
\frac{2}{2\alpha-1} n^{-1}, & \alpha>0.5,
\end{array} 
\right.
\ed
where
\bd
\begin{array}{rcl}
C_{a}(\alpha,\delta) & = & \frac{4\alpha}{1-2\alpha}\Gamma(1+2\alpha)+\delta^{2}(\Gamma(1+2\alpha)-\Gamma^{2}(1+\alpha)),\\
C_{c}(\alpha) & = & \frac{2\alpha\Gamma(2\alpha+2)}{1-2\alpha}.
\end{array}
\ed

Asymptotically as $n\rightarrow \infty$ the correlation coefficient behaves (depending \mbox{
on $\alpha$)} as,
\bd
\rho_{n} \sim
\left\{
\begin{array}{ll}
\left((1-\kappa)C_{a}(\alpha,\delta)+ \kappa C_{c}(\alpha) \right)n^{-2\alpha}, & 0< \alpha<0.5,\\
2n^{-1}\ln n, & \alpha=0.5,\\
\frac{2}{2\alpha-1} n^{-1}, & \alpha>0.5.
\end{array} 
\right.
\ed

Depending on whether $\alpha<0.5\lambda$, $\alpha=0.5\lambda$ or $\alpha>0.5\lambda$
(remember the model equivalency with $\lambda \neq 1$)
there are different asymptotic regimes. This has been also noticed by \citet{RAdaPMil20111,RAdaPMil20112} 
and \citet{,KBarSSag2012}.
An intuitive explanation why for $\alpha<0.5$ a completely different behaviour occurs
can be that in this case the branching rate is relatively high (with respect to $\alpha$) 
and local correlations will dominate over the ergodic properties
of the Ornstein--Uhlenbeck process \citep{RAdaPMil20111,RAdaPMil20112}. However, why this threshold lies at exactly
$\alpha=0.5\lambda$ remains unclear. 

As $\kappa \rightarrow 1$ the correlation coefficient converges to $C^{(n)}_{c}(\alpha)/V^{(n)}_{c}(\alpha)\ge 0$.
With fixed $\alpha$ and $n$ one can immediately see that this has to be
a monotonic, either increasing or decreasing convergence. 
Because a jump component adds independent of the trait value noise to the system
one can expect it to be a decreasing convergence, and plotting
the correlation for different values of the remaining parameters confirms this,
\mbox{(Fig. \ref{figOUcorrK}).} However, a full mathematical proof is still lacking due
to the delicate interactions of the different components of $\rho_{n}(\kappa)$. 
The conjecture stated below gives us the equivalent condition for the 
interspecies correlation coefficient $\rho_{n}(\kappa)$ to
decrease monotonically for \mbox{
$\kappa \in (0,1)$.} It is enough to
consider $\delta=0$, 
as for all $n\ge 2$ and $\alpha, \delta \ge 0$ I have, $V^{(n)}_{c}(\alpha) \ge C^{(n)}_{c}(\alpha)$.
\begin{conjecture}\label{conjRhon}
For all $\alpha \ge 0$, $n\ge 2$ 
\be
V^{(n)}_{c}(\alpha)C^{(n)}_{a}(\alpha,0) \ge V^{(n)}_{a}(\alpha,0) C^{(n)}_{c}(\alpha).
\ee
\end{conjecture}

\begin{figure}
\centering
\includegraphics[width=0.32\textwidth]{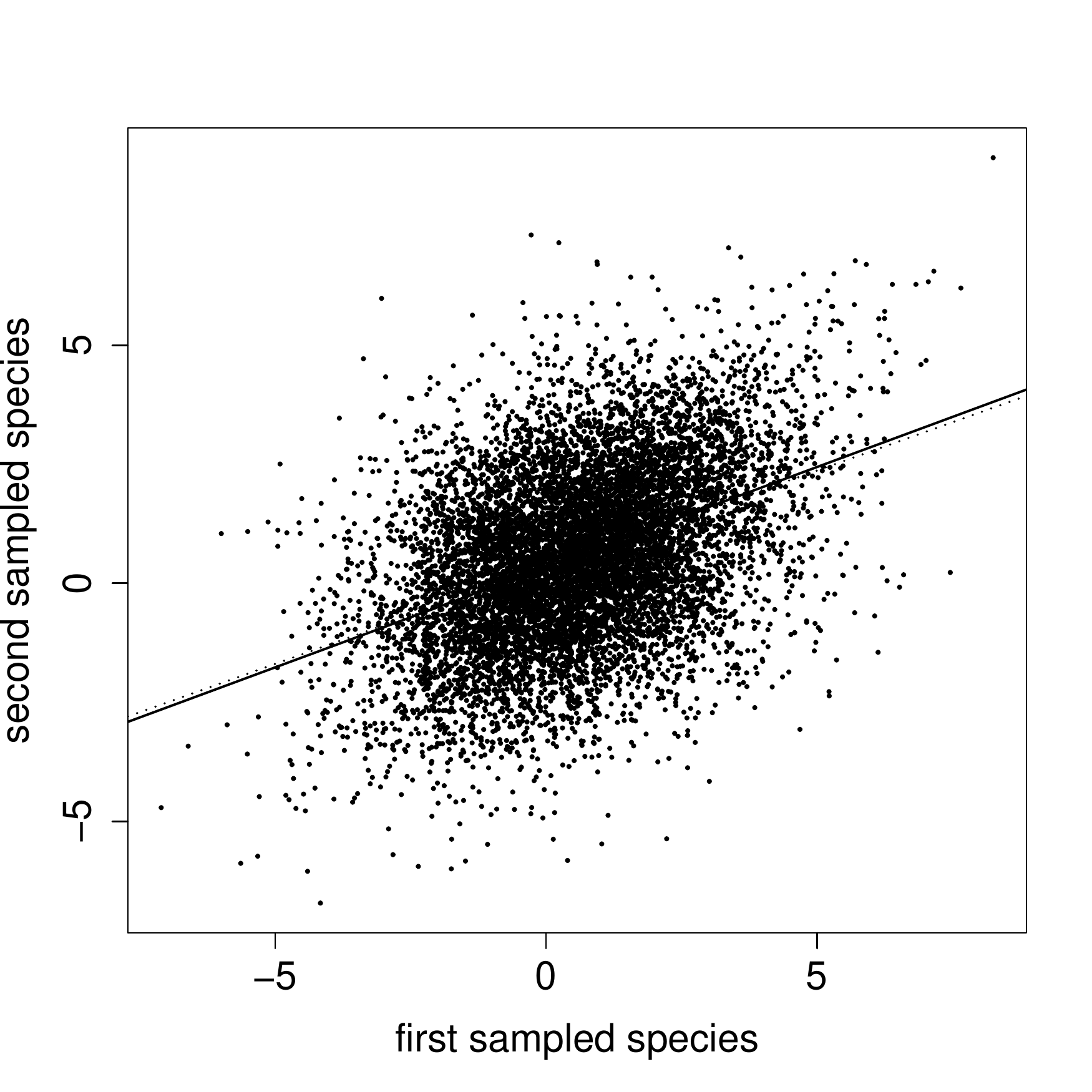}
\includegraphics[width=0.32\textwidth]{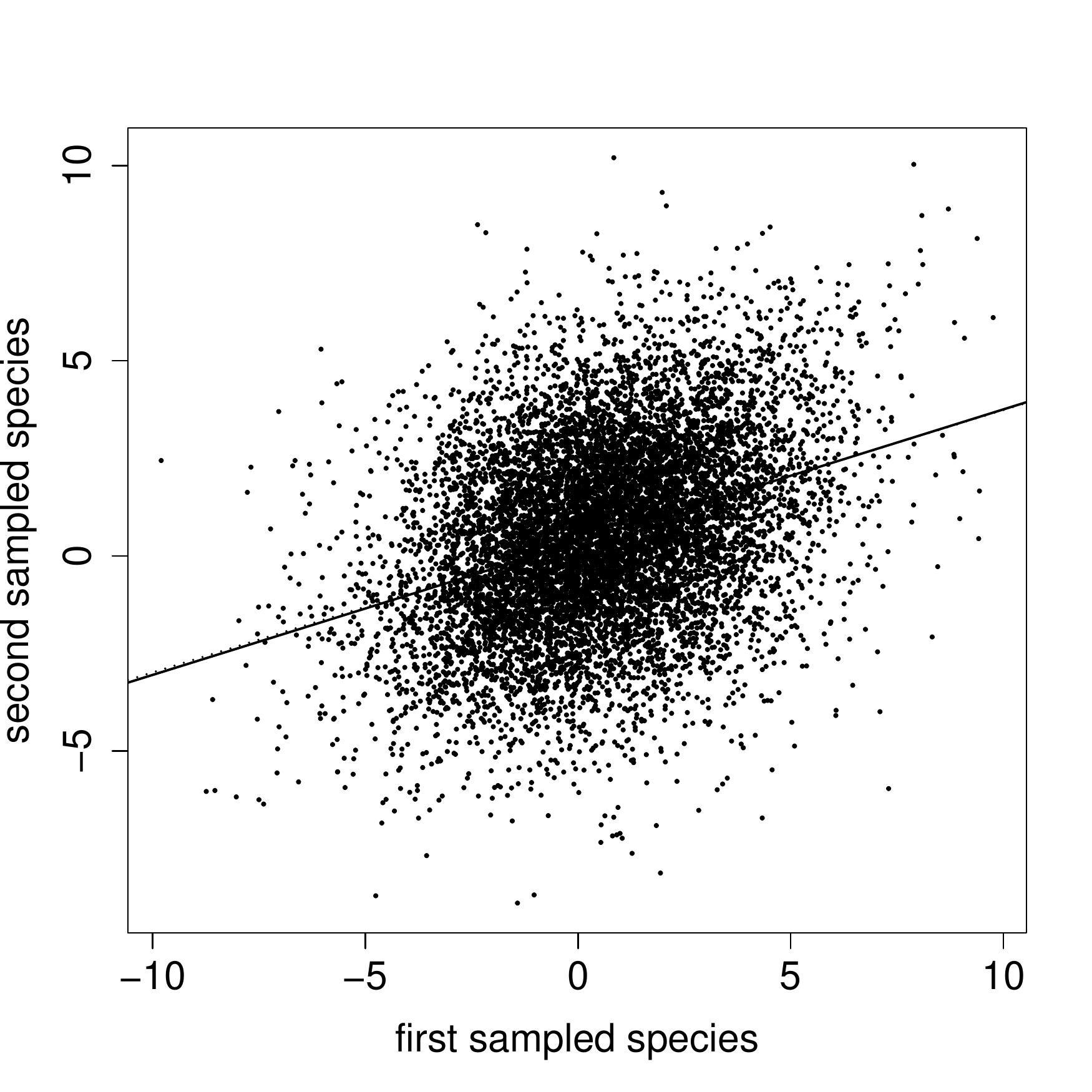}
\includegraphics[width=0.32\textwidth]{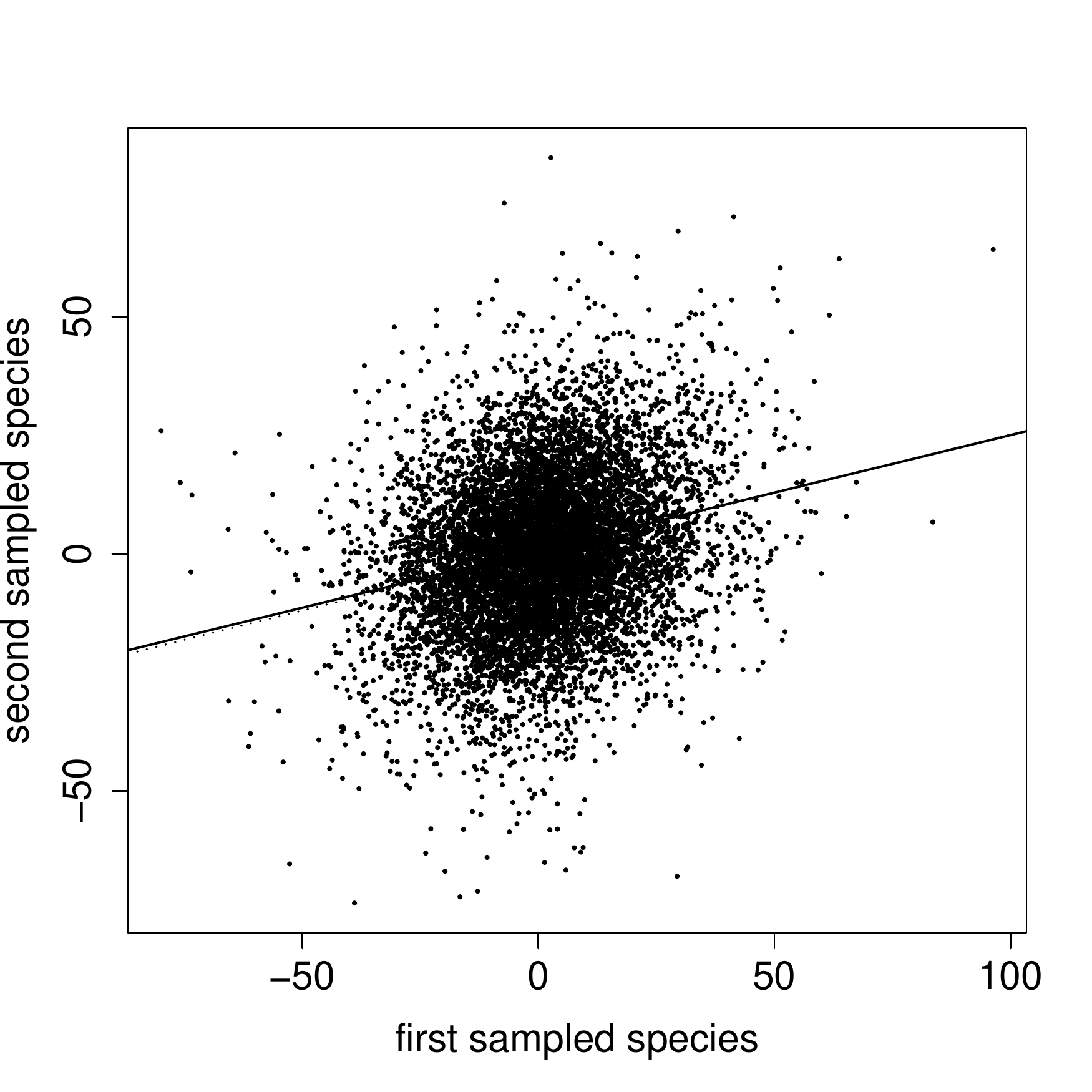} \\
\includegraphics[width=0.32\textwidth]{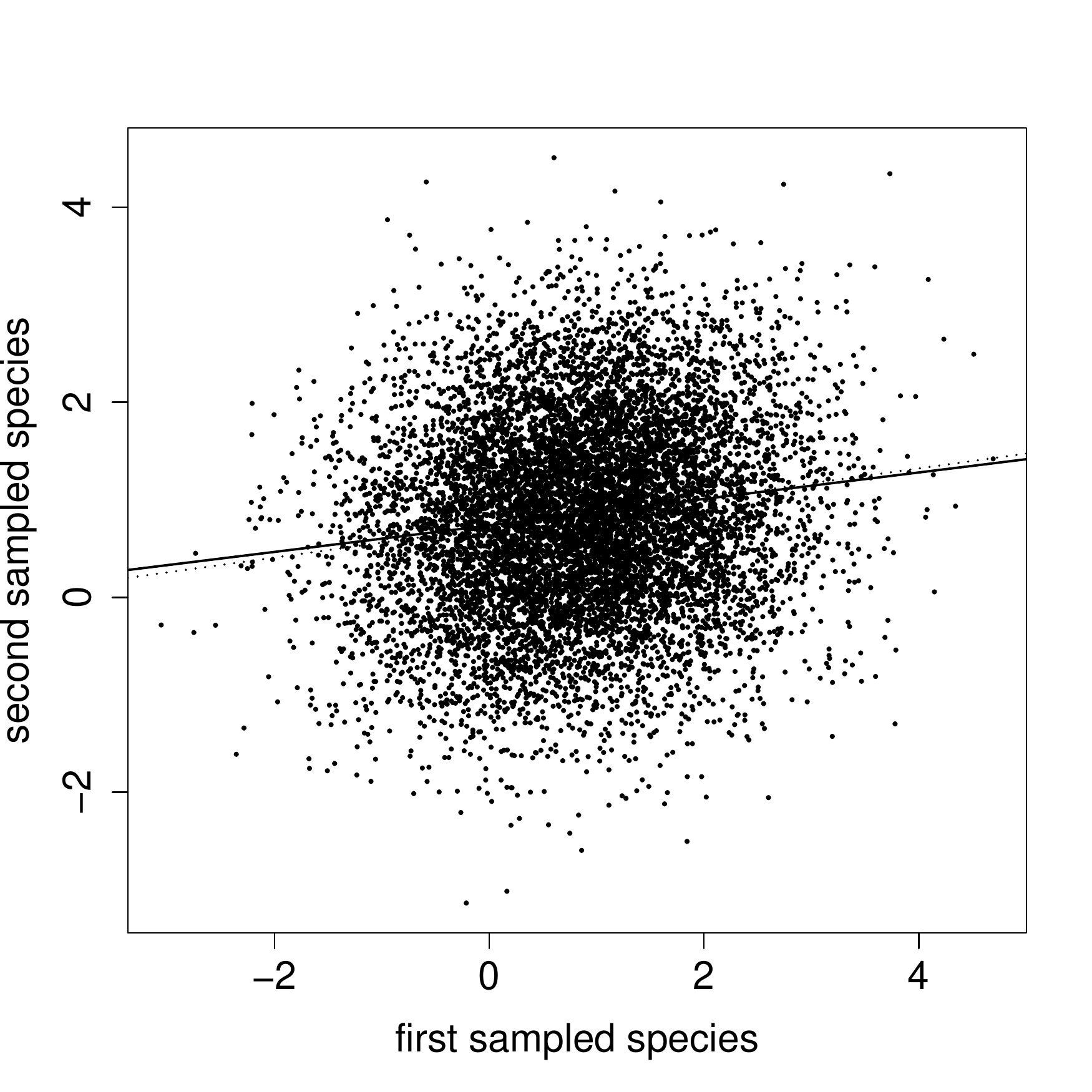}
\includegraphics[width=0.32\textwidth]{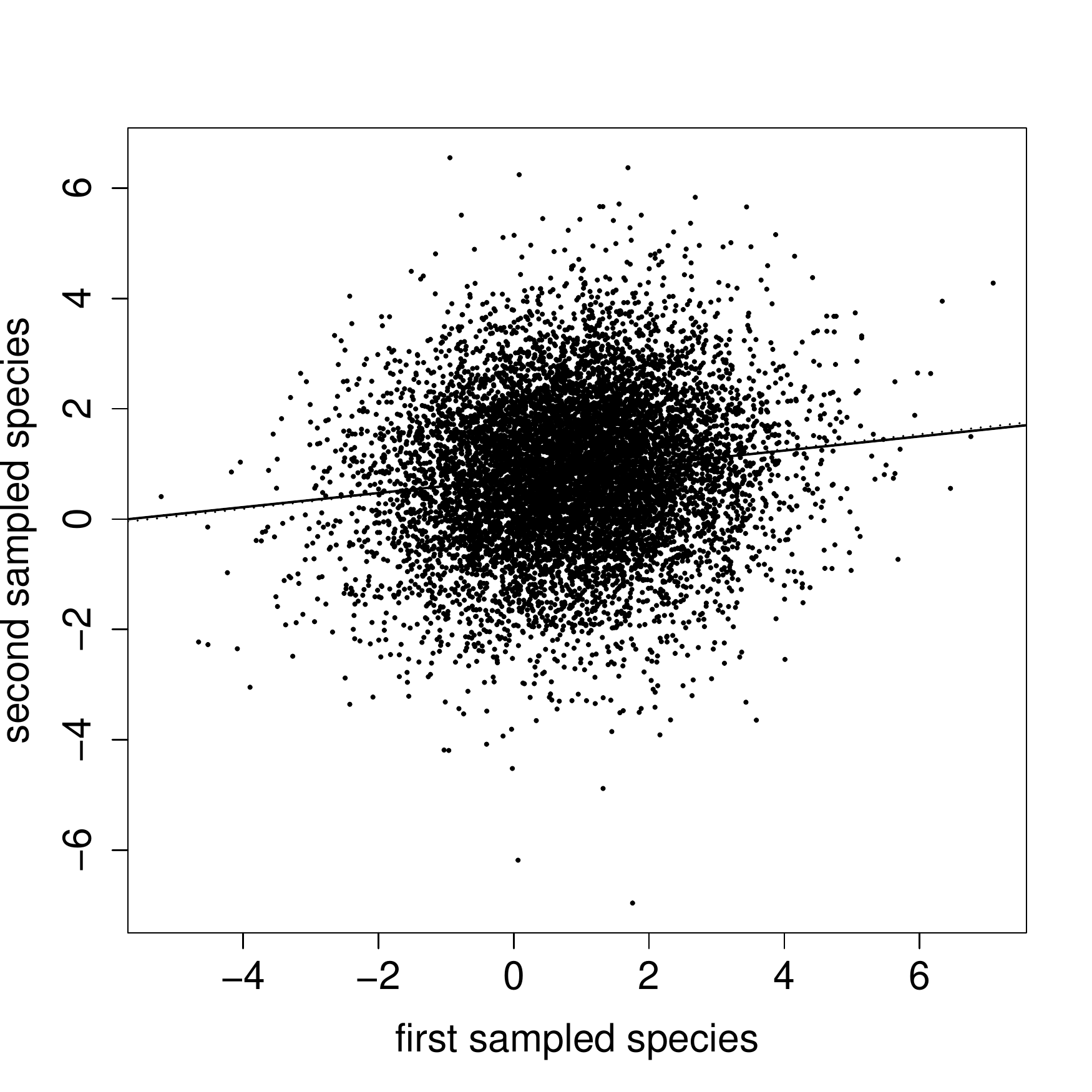}
\includegraphics[width=0.32\textwidth]{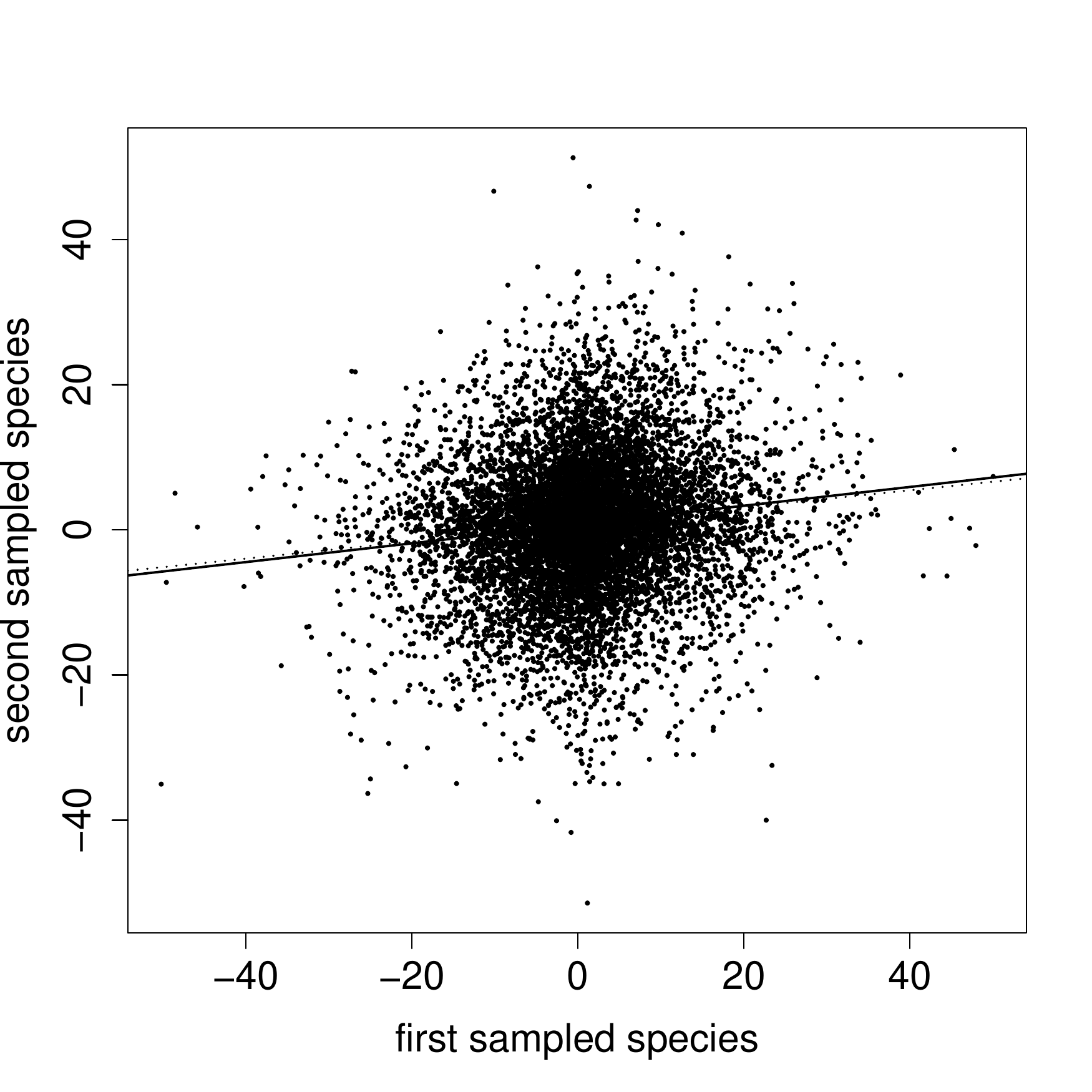} \\
\includegraphics[width=0.32\textwidth]{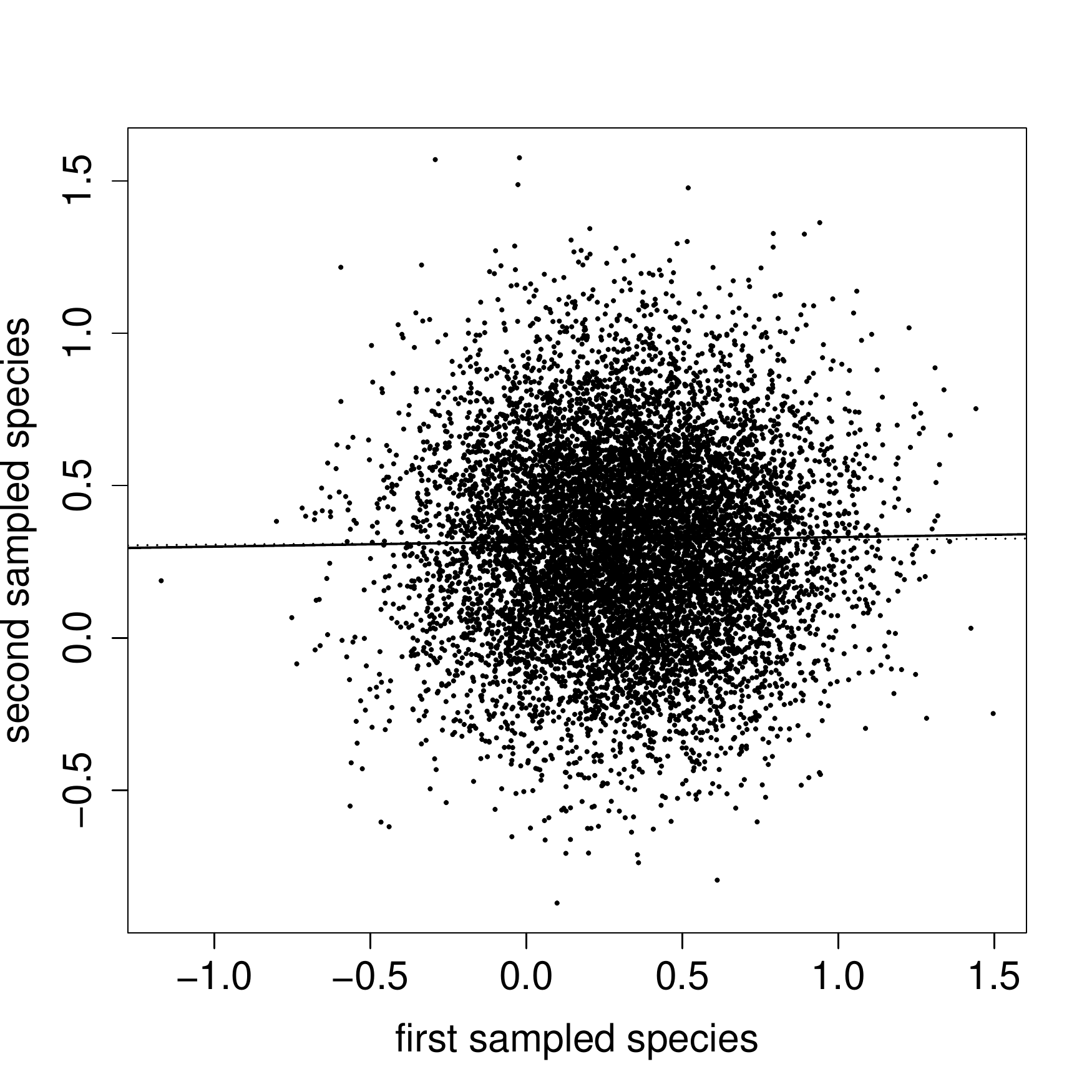} 
\includegraphics[width=0.32\textwidth]{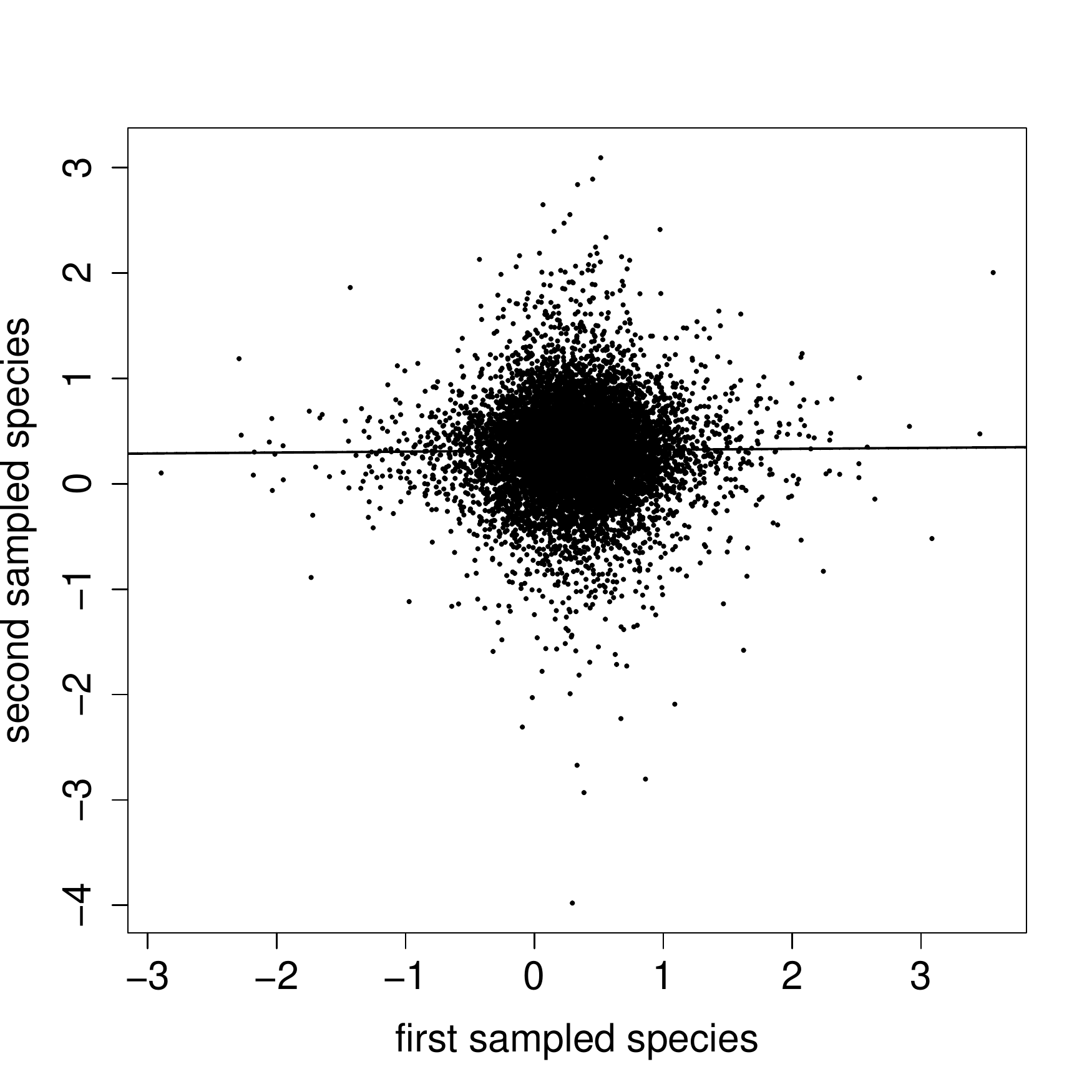}
\includegraphics[width=0.32\textwidth]{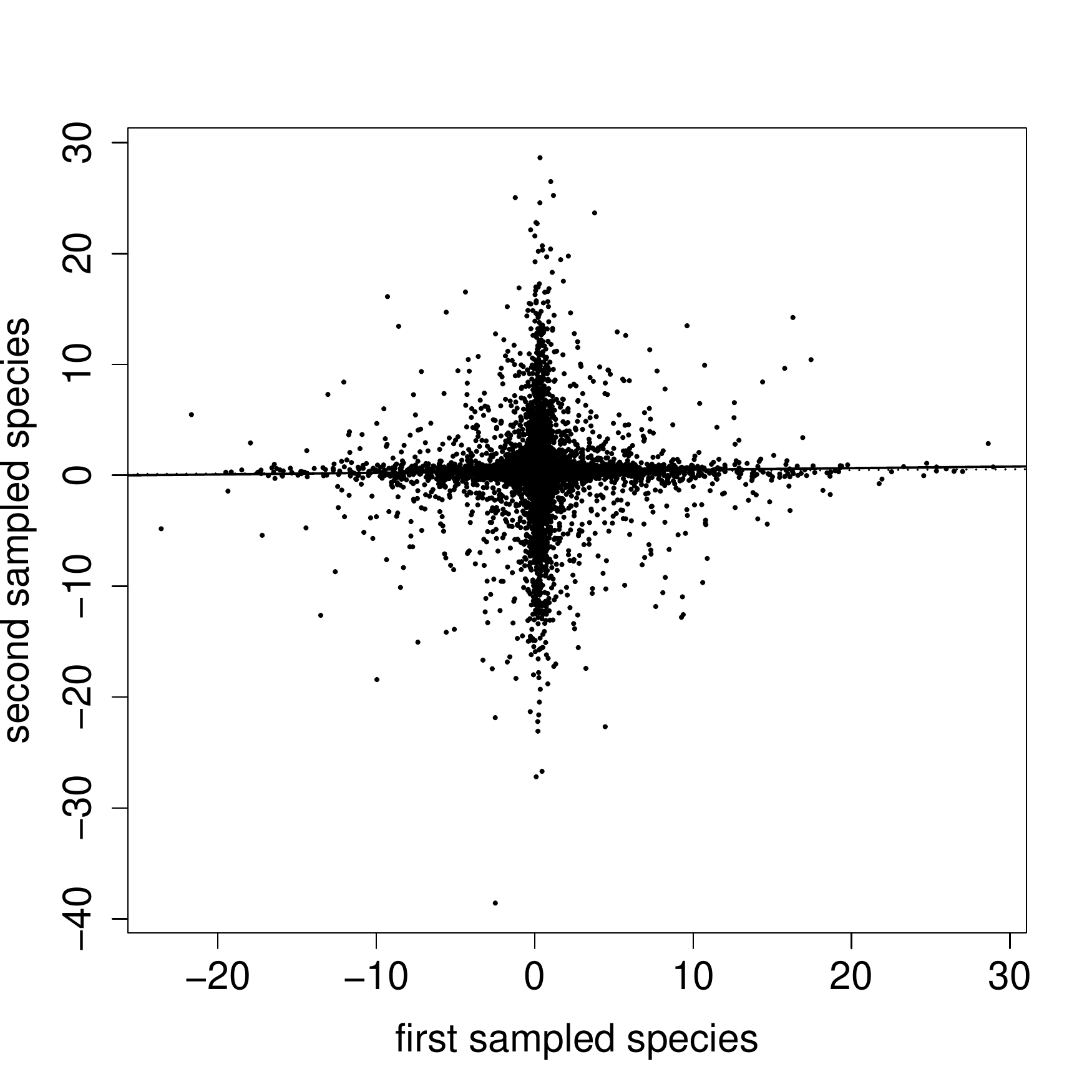}
\caption{
Regression lines fitted to simulated data (thick line) mostly indistinguishable from the true regression line
\mbox{
$y=\rho_{n}x+(1-\rho_{n})(b_{n,\alpha}X_{0}+(1-b_{n,\alpha})\theta)$} (dotted line) with $\rho_{n}$ given by the exact formula for different
values of $\kappa$, $\alpha$ and $\delta$ in the Yule--Ornstein--Uhlenbeck--jumps model. 
Top row $\alpha=0.05$, center row $\alpha=0.5$, bottom row $\alpha=5$. First column $\kappa=0.01$
second column $\kappa=0.5$ and third column $\kappa=0.99$. The other parameters are fixed at $\delta=1$, $X_{0}=0$, $\sigma_{a}^{2}=1$, $p=0.5$,
so $\sigma_{c}^{2}=\kappa/(1-\kappa)$.}
\end{figure}

\begin{figure}
\centering
\includegraphics[width=0.4\textwidth]{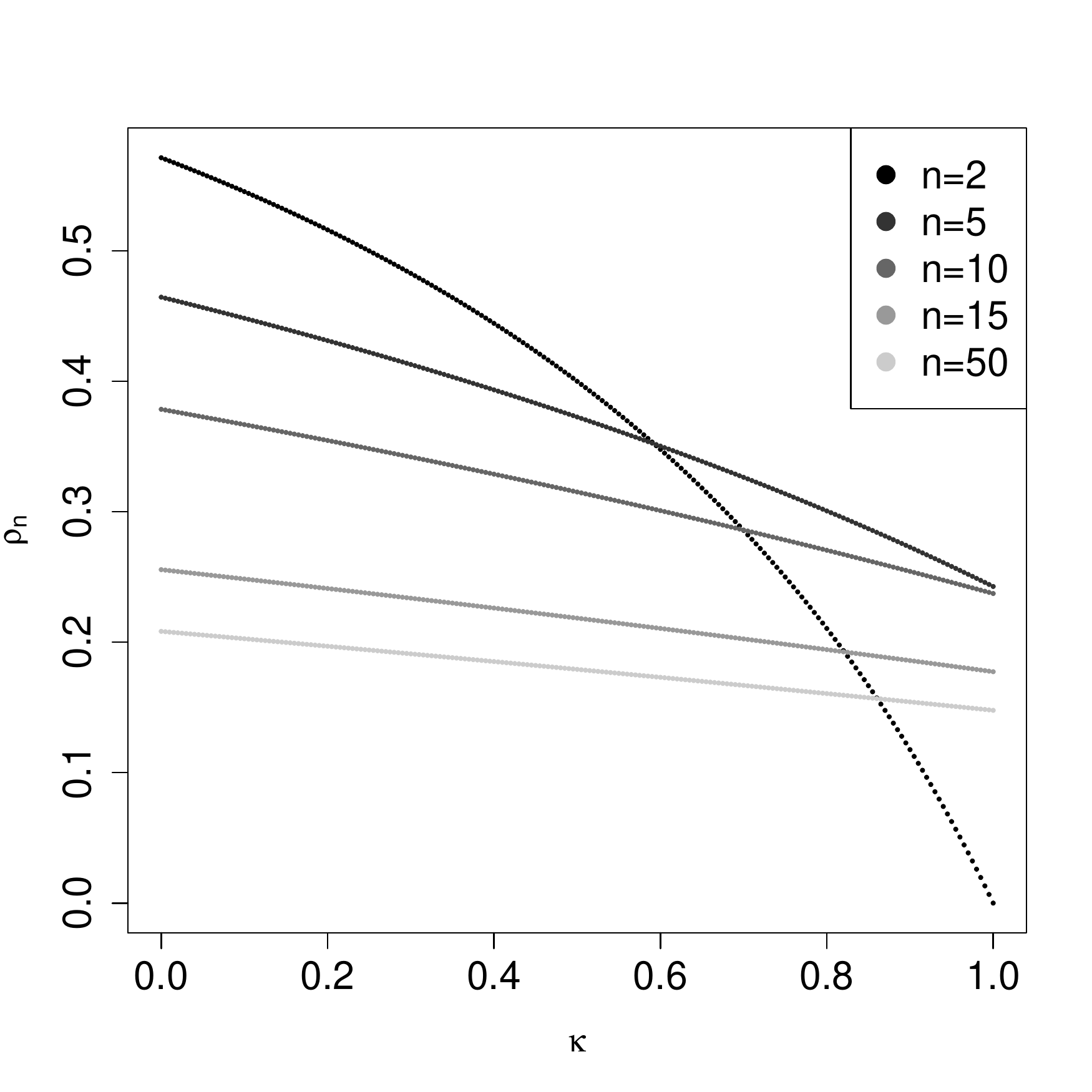}
\includegraphics[width=0.4\textwidth]{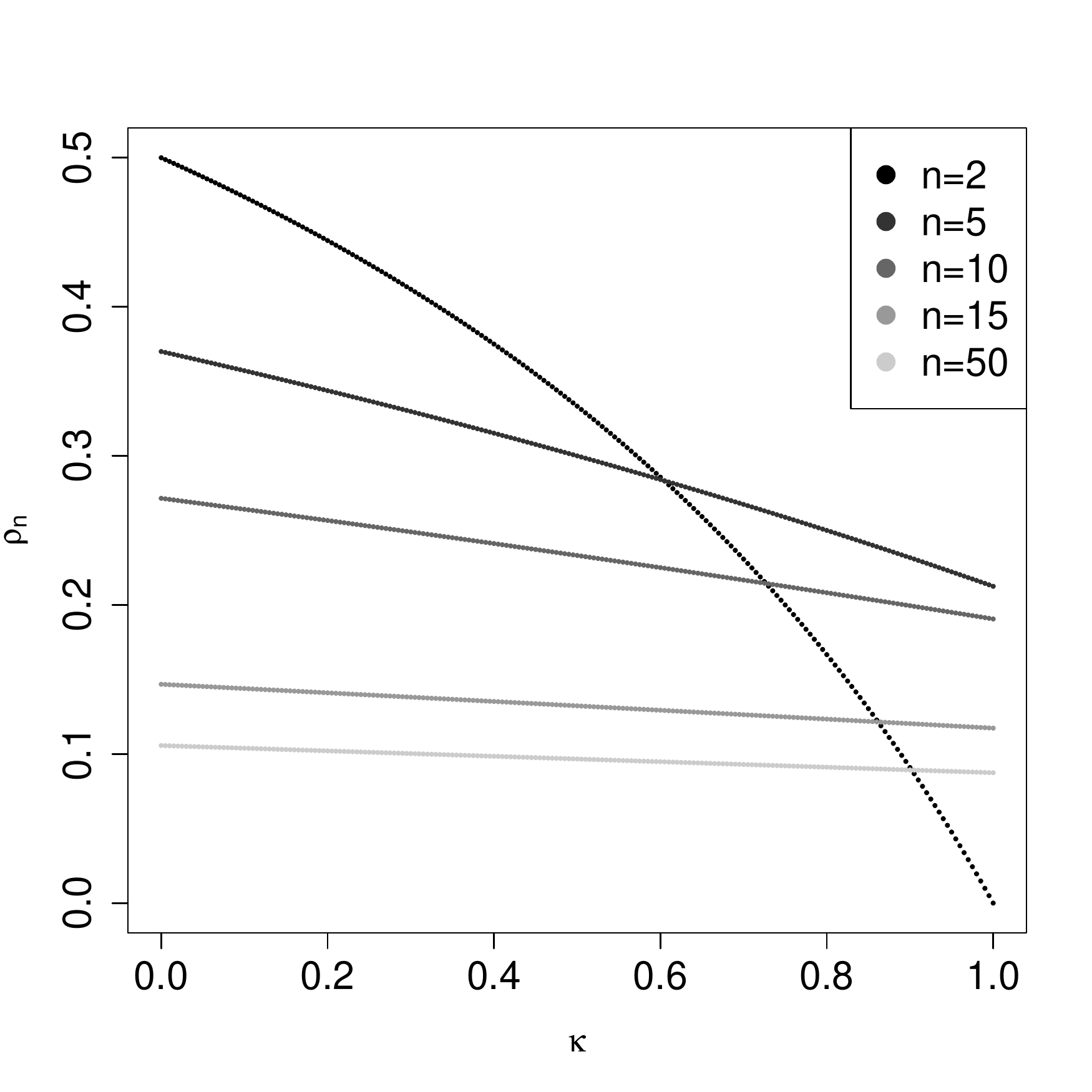} \\
\includegraphics[width=0.4\textwidth]{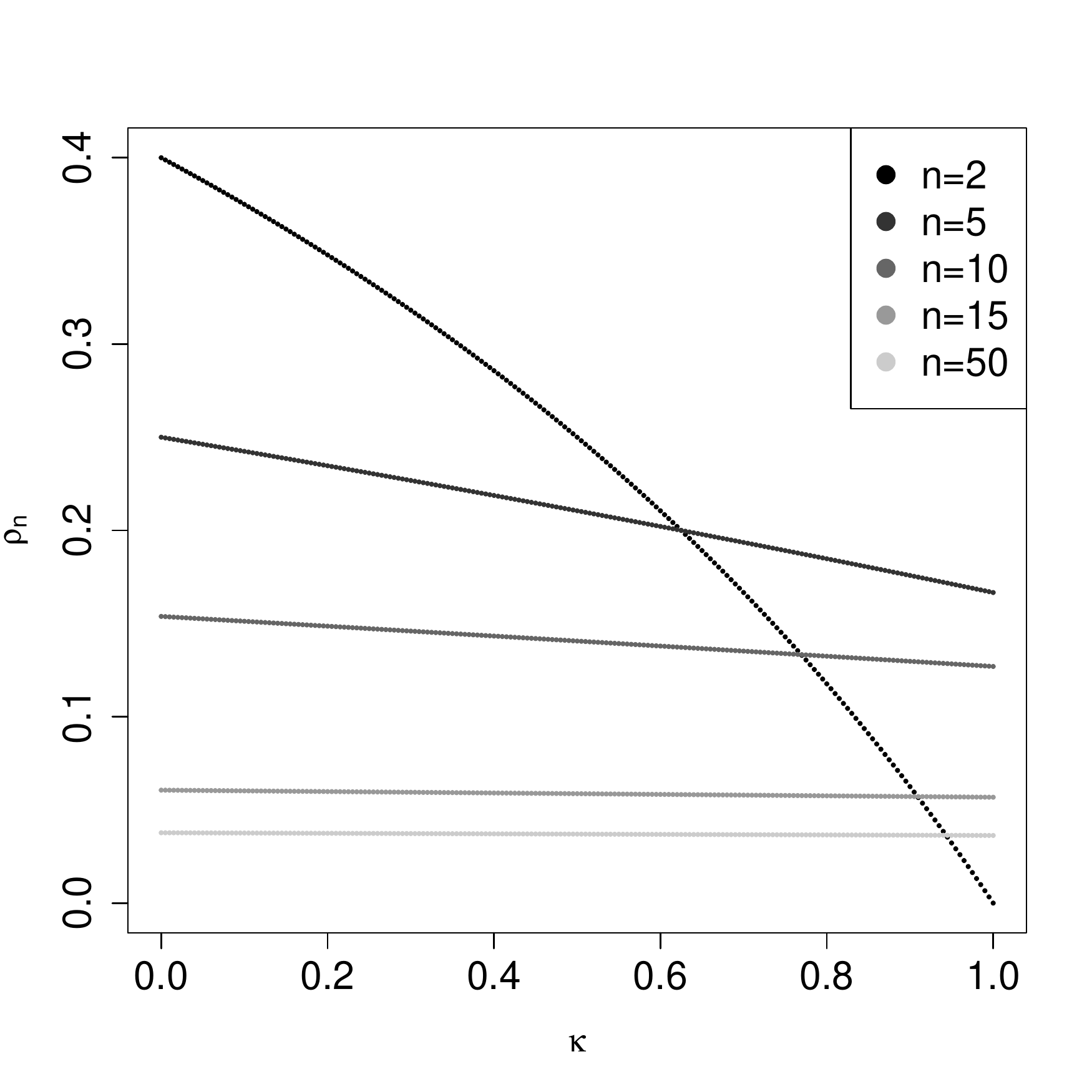}
\includegraphics[width=0.4\textwidth]{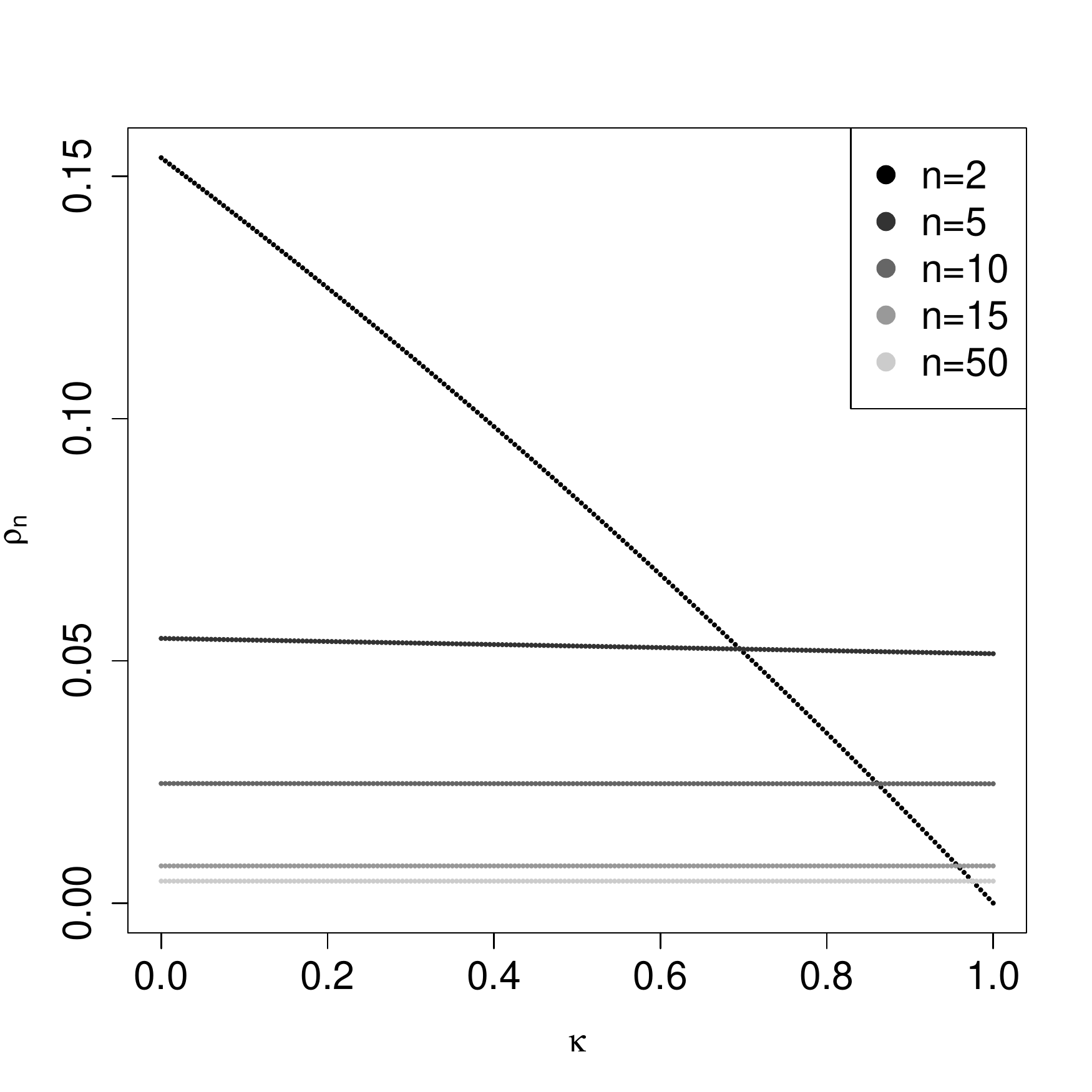} 
\caption{Interspecies correlation coefficient for the Yule--Ornstein-Uhlenbeck--jumps model for $\delta=0$ 
and different values of 
$\alpha$, $n$. Top left: $\alpha=0.25$, top right: $\alpha=0.5$, bottom left: $\alpha=1$, bottom right $\alpha=5$.
}\label{figOUcorrK}
\end{figure}

\subsection{Introducing extinction}
Above I concentrated on the case of pure birth trees. However, a more general version
of the derived formulae can be used to include death events. For the unconstrained 
evolutionary model i.e. Brownian motion, I found that:
\be
\begin{array}{rcl}
\var{X} & = & (\sigma_{a}^{2}+2p\sigma_{c}^{2})\left((1-\kappa)\E{T} + \kappa \frac{1}{2}\E{\Upsilon}\right), \\
\cov{X}{Y} & = &  (\sigma_{a}^{2}+2p\sigma_{c}^{2})\left((1-\kappa)\E{T-\tau} + \kappa \frac{1}{2}\E{\upsilon}\right), \\
\rho_{n} & = & \frac{(1-\kappa)\E{T-\tau} + \kappa \frac{1}{2}\E{\upsilon}}{(1-\kappa)\E{T} + \kappa \frac{1}{2}\E{\Upsilon}}
\end{array}
\ee
and for the constrained model, with Ornstein--Uhlenbeck dynamics:
\be
\begin{array}{rcl}
\var{X} & = &  \frac{\sigma_{a}^{2}+2p\sigma_{c}^{2}}{2\alpha}\left((1-\kappa)V^{(n)}_{a}(\alpha,\delta) 
+ \kappa V^{(n)}_{c}(\alpha) \right),\\
\cov{X}{Y} & = &  
\frac{\sigma_{a}^{2}+2p\sigma_{c}^{2}}{2\alpha}\left((1-\kappa)C^{(n)}_{a}(\alpha,\delta) 
+ \kappa C^{(n)}_{c}(\alpha) \right),
 \\
\rho_{n} & = & \frac{(1-\kappa) C^{(n)}_{a}(\alpha,\delta) + \kappa C^{(n)}_{c}(\alpha)}{
(1-\kappa)V^{(n)}_{a}(\alpha,\delta) + \kappa V^{(n)}_{c}(\alpha)},
\end{array}
\ee
where $p$ is the probability of a jump occurring and,
\bd
\begin{array}{rcl}
V^{(n)}_{a}(\alpha,\delta) &=& 1-\E{e^{-2\alpha T}} +\delta^{2}\var{e^{-\alpha T}}, \\
V^{(n)}_{c}(\alpha) &=& \frac{2\alpha}{2}\E{\sum\limits_{i=2}^{\Upsilon+1}e^{-2\alpha(t_{\Upsilon+1}+\ldots+t_{i})}} ,\\
C^{(n)}_{a}(\alpha,\delta)& =& \E{e^{-2\alpha \tau}}-\E{e^{-2\alpha T}} +\delta^{2}\var{e^{-\alpha T}}, \\
C^{(n)}_{c}(\alpha)&=& \frac{2\alpha}{2}\E{\sum\limits_{i=2}^{\upsilon+1}e^{-2\alpha(\tau+t_{\upsilon+1}+\ldots+t_{i})}},
\end{array}
\ed
where $(t_{1},\ldots,t_{\Upsilon+1})$ are the times between speciation events on a randomly chosen lineage, see 
\mbox{Fig. \ref{figappTreeTimes}.}

The values of $\E{T}$ and $\E{T-\tau}$ for a birth--death tree with constant coefficients $\lambda\ge \mu >0$
where discussed by \citet{SSagKBar2012} \citep[see also][]{TGer2008a,TGer2008b,AMooetal,TSta2008,TStaMSte2012}. 
The formulae for the Laplace transforms of $T$ and $\tau$
and also for the expectation and probability generating functions of $\Upsilon$
and $\upsilon$ are to the best of my knowledge not available yet
\citep[but see][for some distributional properties]{AMooetal,TStaMSte2012}, however
they could be obtained via simulation methods \citep[using e.g. the TreeSim R package][]{TreeSim1,TreeSim2}.
When discussing the Hominoid body--size analysis of \citet{FBok2002} below I will illustrate this approach.

\subsection{Quantifying effects of gradual and punctuated change --- quadratic variation}
In the previous section I described first-- and second--order properties
(mean, variance, covariance, correlation) of an evolutionary model
containing both gradual and punctuated change. 
One may ask what is the main effect of cladogenetic evolution on
the phylogenetic sample
and also how one can elegantly summarize the magnitude of its effect.
From the sample's point of view one can see (if Conjecture \ref{conjRhon} is correct) 
that increasing the variance
of the jump with respect to the stochastic perturbation of the phenotype 
will decorrelate the contemporary observations, i.e. they will be less
similar than expected from gradual change. At first glance this
might seem like a ``nearly'' obvious statement but notice that 
the covariance is also increased due to the jump component and 
all jumps,
apart from those on pendant branches, are shared by some subclade
of species. Therefore a large jump early enough (to be shared 
by a large enough subclade) could ``give similarity'' to these tips.

The decrease of similarity between species with $\kappa$ (conditional on Conjecture \ref{conjRhon}) is the main effect
of cladogenetic change but it would be biologically useful to have some value
characterizing the magnitude of its effect. 
One such possibility is
the expectation of the quadratic variation (see Appendix B for a
brief mathematical introduction to this) of the evolutionary process along a single lineage.
As one can attach 
mechanistic interpretations to the different modes of evolution \citep{AMooSVamDSch1999}
appropriately partitioning the reasons for evolution could give insights into
the trait's role.

The two models of anagenetic evolution considered here have the same quadratic variation $\sigma^{2}_{a}t$. 
The cladogenetic component comes in as a mean $0$,
variance $\sigma_{c}^{2} < \infty$ jump just after a speciation event added  
with probability $p$. 
Therefore for both of these models (and in fact for any one defined by Eq. \ref{eqSDEdiff})
with the additional jump component the quadratic variation of a given lineage will be
\begin{equation}\label{eqQVtriat}
[X] = \sigma_{a}^{2}T + \sigma_{c}^{2}\Upsilon^{\ast}.
\end{equation}
One can clearly see that the magnitude of the trait's fluctuations can be divided into the
cladogenetic and anagenetic component and they are defined respectively by the parameters $\sigma_{c}^{2}$
and $\sigma_{a}^{2}$ \citep[as considered in][]{FBok2002}. However, the formula for the quadratic
variation depends on $T$ and $\Upsilon^{\ast}$. These are random unless the tree is given and the
jump pattern known. Therefore this is a random variable
and instead I consider its expectation,
\begin{equation}
\E{[X]} = \sigma_{a}^{2}\E{T} + p\sigma_{c}^{2}\E{\Upsilon}.\label{eqQvarJ}
\end{equation}
With a given phylogeny one might consider that using $\E{T}$ is superfluous but this could be useful
if one wants to make predictive statements not dependent on the given phylogeny.
Under the considered conditioned Yule tree I will have,
\begin{equation}
\E{[X]} = \sigma_{a}^{2}H_{n} + 2p\sigma_{c}^{2}(H_{n}-1)\sim (\sigma_{a}^{2}+2p\sigma_{c}^{2})\ln n.
\end{equation}
One way of quantifying the proportion of cladogenetic change is using the value
$$\frac{2p\sigma_{c}^{2}(H_{n}-1)}{\sigma_{a}^{2}H_{n} + 2p\sigma_{c}^{2}(H_{n}-1)}$$
(and respectively $$\frac{\sigma_{a}^{2}H_{n}}{\sigma_{a}^{2}H_{n} + 2p\sigma_{c}^{2}(H_{n}-1)}$$ for the anagenetic effects). 
For a large enough set of tip species the above would simplify to $\kappa$ (and $1-\kappa$). 

\citet{TMatFBok2008} proposed two other ways of comparing anagenetic and cladogenetic effects
for Brownian motion punctuated evolution.
The first one is less relevant here as it involves an estimation
procedure. One sets $\sigma_{c}^{2}=0$ and estimates $\sigma_{a}^{2}$ under this assumption.
Next one estimates both parameters and compares by how much has the estimate of $\sigma_{a}^{2}$
dropped. The second is very similar to the proposed quadratic variation approach. 
Knowing the speciation and extinction rates the 
average species lifetime or average time to the next split is known
\citep{FBokVBriTSta2012,AMooetal,TStaMSte2012,MSteAMoo2010}.
In the conditioned Yule tree framework used here the length, $t_{b}$, of a
random interior edge is exponential with rate $2\lambda$ and hence will have
expected length $1/2\lambda$
\citep{AMooetal,TStaMSte2012,MSteAMoo2010}.
Therefore between speciation points the percentage of evolution due to the gradual component will be,
\be \label{eqBokmaProp}
\frac{\sigma_{a}^{2}\E{t_{b}}}{\sigma_{a}^{2}\E{S}+\sigma_{c}^{2}} = 
\frac{0.5\lambda\sigma_{a}^{2}}{0.5\lambda\sigma_{a}^{2}+\sigma_{c}^{2}} 
=\frac{\sigma_{a}^{2}}{\sigma_{a}^{2}+2\lambda\sigma_{c}^{2}}.
\ee
Notice that this will correspond exactly to the $1-\kappa$ limit derived in the
previous paragraph after one small modification. Namely I assume that 
the tree starts at the moment of the first speciation event, i.e.
disregarding the root branch. Then $\E{T}=H_{n}-1$ and for all
$n$ the proportion of anagenetic change will be $1-\kappa$ 
(and cladogenetic $\kappa$). Taking $p=1$ i.e. a jump occurred 
at each speciation event I arrive at Eq. \eqref{eqBokmaProp}.

These two procedures were described by \citet{TMatFBok2008} for a Brownian motion model
of gradual evolution. However, the discussion concerning quadratic variation
above suggests that 
they are valid for any model described by Eq. \eqref{eqSDEdiff}.

\section{Interpreting Hominoid body size evolution results}
\label{secHom}
For illustrative purpose I will now discuss how the conclusions of \citet{FBok2002} look under my approach. 
I realize that this clade is not the most appropriate one to be analyzed by a tree--free model
because its phylogeny is available. I nevertheless chose to use it here for a number of reasons. 
First birth and death rate parameters and also the rates of 
anagenetic and cladogenetic evolution
have already been preestimated. I mentioned in the introduction that I do not consider
estimation of parameters in this current work but want to concentrate on understanding their effects
and using a clade with a known phylogeny makes this easier. One can also use my results
if only the time to the most recent common ancestor of the clade is known and I discuss such
a situation here.

I do not claim that what I write here about Hominoid body size evolution 
is authoritative but rather wish to motivate the usefulness and applicability
of the methods presented above. 
In the subsequent discussion I assume that the unit of time is 1 million years.

\citet{FBok2002} studied the evolution of (the logarithm of) Hominoid body size under a Brownian motion
model of gradual evolution \mbox{(diffusion coefficient 
$\sigma_{a}^{2}$)} with an independent normally distributed mean $0$, variance
$\sigma_{c}^{2}$ jump added to a newborn daughter lineage at the speciation instant. 
This is equivalent to the presented here model with $p=0.5$ 
\citep[``We can only state in retrospect that the probability that a speciation event has affected the phenotype of interest is \nicefrac{1}{2}.''][]{FBok2002}.

\citet[][see Fig.2 therein]{FBok2002} considers the Hominoidea phylogeny of \citet{APur1995} consisting of five species,
\textit{Gorilla gorilla}, \textit{Homo sapiens}, \textit{Pongo Pygmaeus}, \textit{Pan paniscus}
and \textit{Pan troglodytes}. \citet{APurSNeePHar1995} estimated a birth rate of $\lambda=0.134$
and death rate $\mu=0.037$. Other rates could be more appropriate
as \citet{FBokVBriTSta2012} obtained $\lambda=0.46~ (95\%~\mathrm{CI}~0.12-1.37)$ 
and $\mu=0.43~(95\%~\mathrm{CI}~0.12-1.37)$ for a seven species phylogeny 
(\textit{Pan paniscus}, \textit{Pan troglodytes}, \textit{Homo sapiens},  \textit{Gorilla gorilla}, \textit{Gorilla beringei}, \textit{Pongo abelii},  \textit{Pongo pygmaeus}).

Using the phylogeny of \citet{APur1995}, \citet{FBok2002} obtained parameter estimates,
$\sigma_{a}^{2}=0.0484$ and $\sigma_{c}^{2}=0.0169$. However, a $95\%$ confidence interval for
$\sigma_{c}^{2}$ contained $0$ but this, as \citet{FBok2002} pointed out, could be caused 
by a lack of statistical power due to the low number of extant species. 

I derived $\E{\Upsilon}$ and $\E{\upsilon}$ only for the pure birth process and so
Eq. \eqref{eqCorYuleBM} is known only in such a situation. 
However, in this specific analysis with known (preestimated) $\lambda$, $\mu$
and fixed number of extant nodes I can obtain estimates of $\E{\Upsilon}$ and $\E{\upsilon}$
via simulations using the TreeSim \citep{TreeSim1,TreeSim2} R package. In \mbox{Fig. \ref{figBDNunu}}
one can see the estimates of $\E{\Upsilon}$ and $\E{\upsilon}$ for different values of $n$.
I calculate the values of $\E{T}$ and $\E{T-\tau}$ 
according to \citet{SSagKBar2012} :
\be
\begin{array}{rcl}
\E{T} & = &\frac{1}{\mu(\gamma-1)}\left(H_{n}+e_{n,\gamma}-\ln\frac{\gamma}{\gamma-1} \right), \\ 
\E{T-\tau} & = & \frac{2}{\mu(n-1)(\gamma-1)}\left(n+ne_{n,\gamma}-\frac{\gamma}{\gamma-1}\left(H_{n}+e_{n,\gamma}-\ln\frac{\gamma}{\gamma-1} \right) \right), 
\end{array}
\ee
where $\gamma=\lambda/\mu$, $e_{n,\gamma}=\int\limits_{0}^{1}\frac{x^{n}}{\gamma-x}$ 
and present the calculated values in Tab. \ref{tabResults}.
From
the formula for 
the expected quadratic variation,
\be
\E{[X]} = \sigma_{a}^{2} \E{T} + \sigma_{c}^{2} \E{\Upsilon^{\ast}} \\ 
\ee
I quantified
the amount of change attributed to anagenetic and cladogenetic evolution by
$$\frac{\sigma_{a}^{2} \E{T}}{\sigma_{a}^{2} \E{T} + \sigma_{c}^{2} \E{\Upsilon}}$$ \\ 
and
$$\frac{\sigma_{c}^{2} \E{\Upsilon}}{\sigma_{a}^{2} \E{T} + \sigma_{c}^{2} \E{\Upsilon}}$$ 
respectively.
\mbox{Table \ref{tabResults}} summarizes these values for the studied tree setups. 


\begin{figure}
\centering
\includegraphics[width=0.4\textwidth]{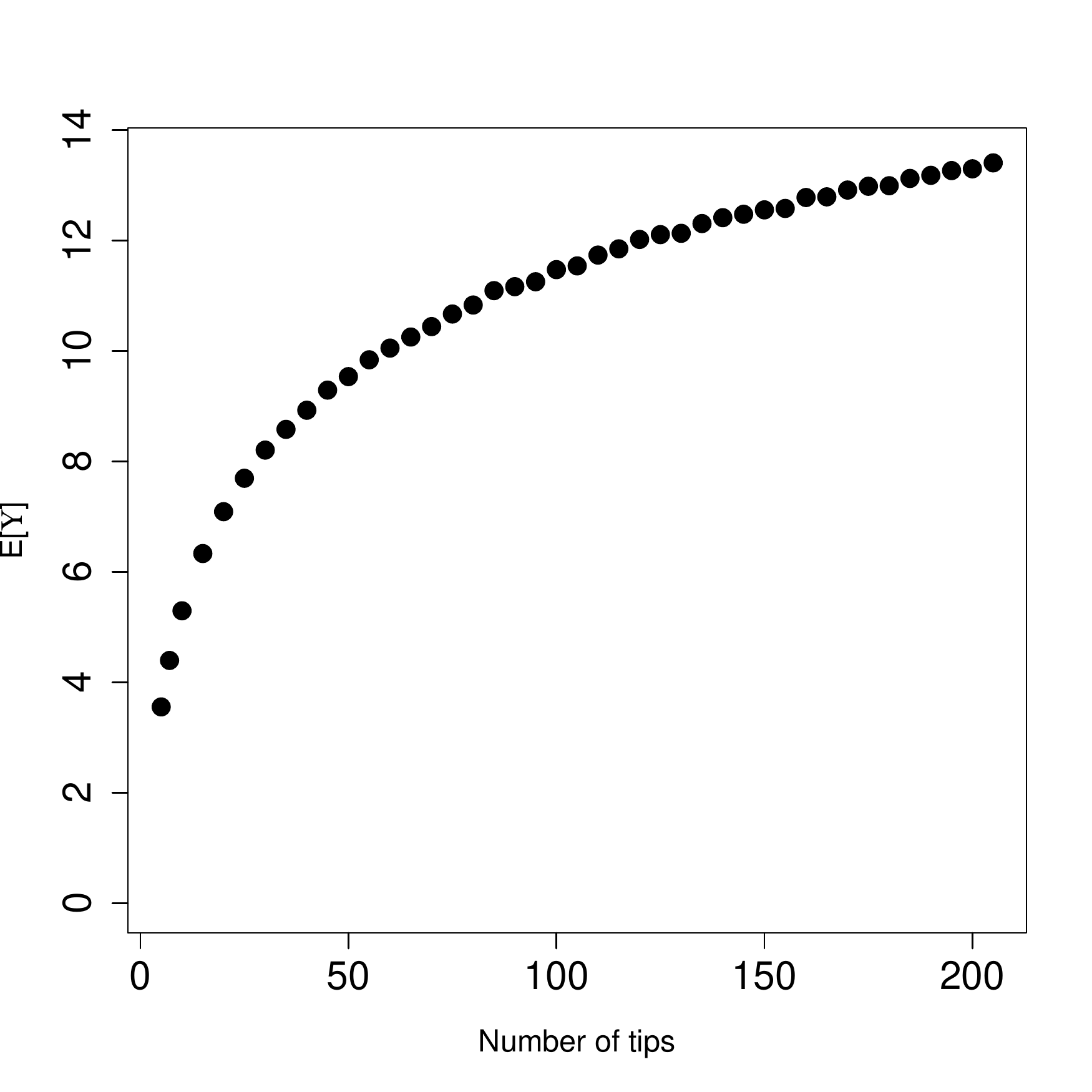}
\includegraphics[width=0.4\textwidth]{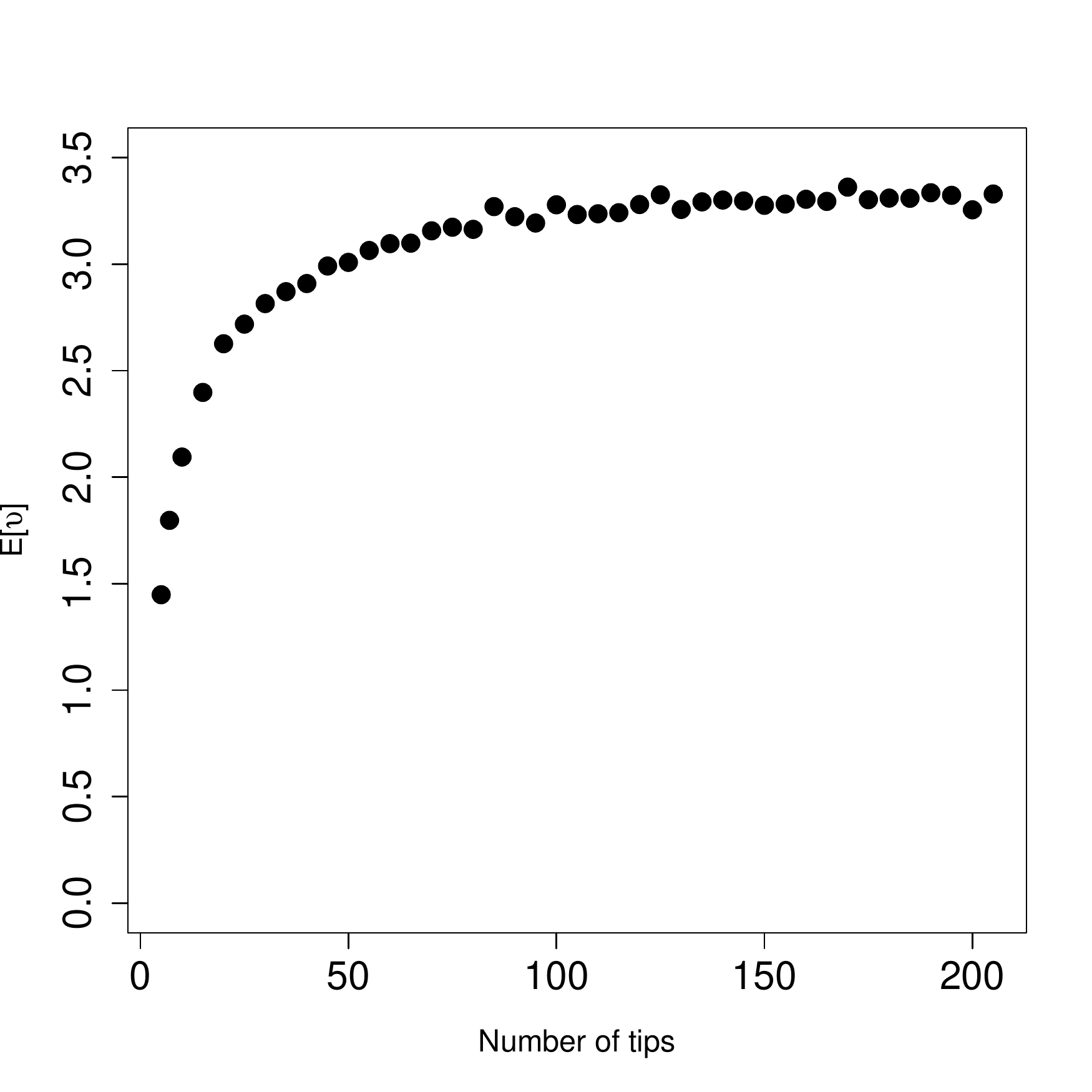}  
\\
\includegraphics[width=0.4\textwidth]{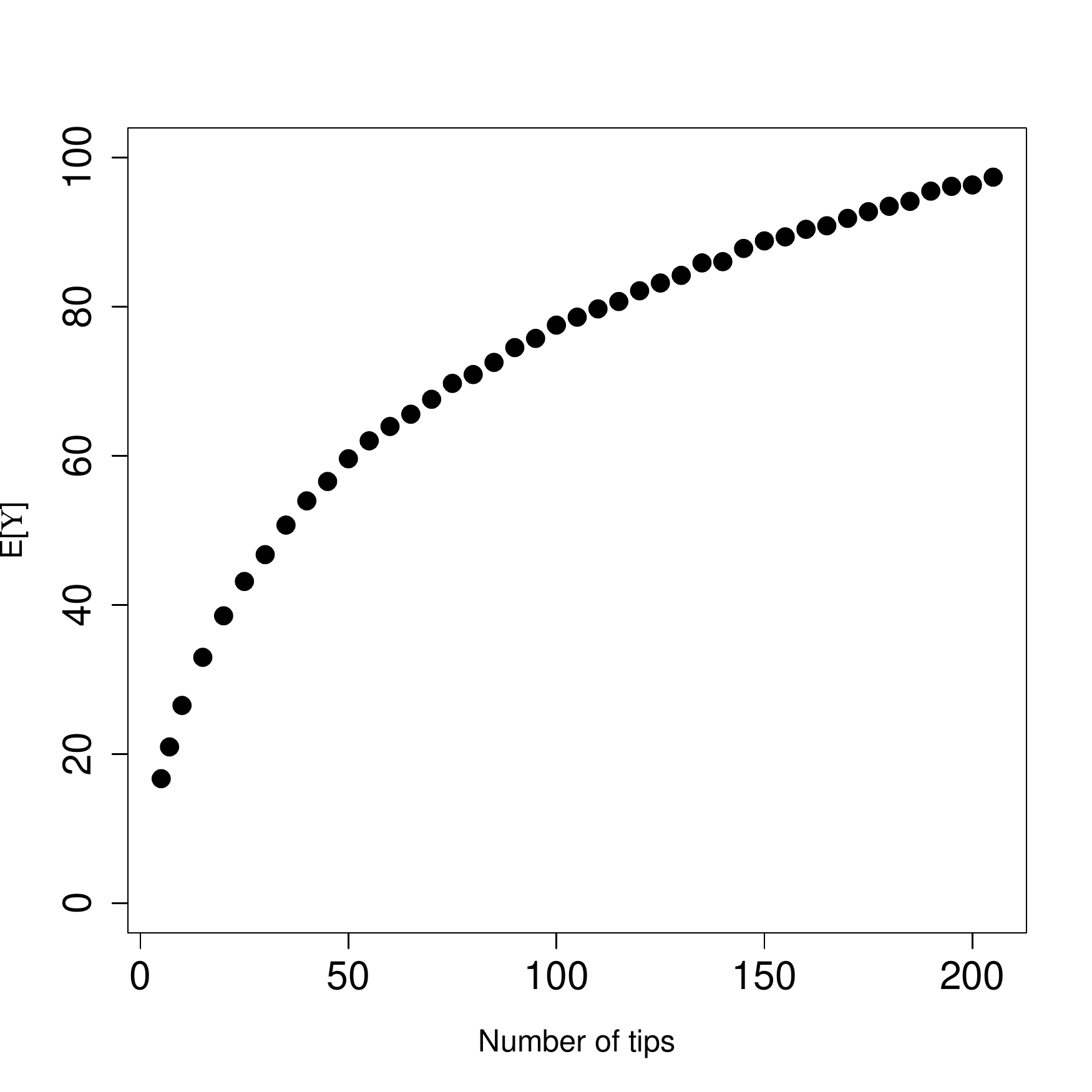}
\includegraphics[width=0.4\textwidth]{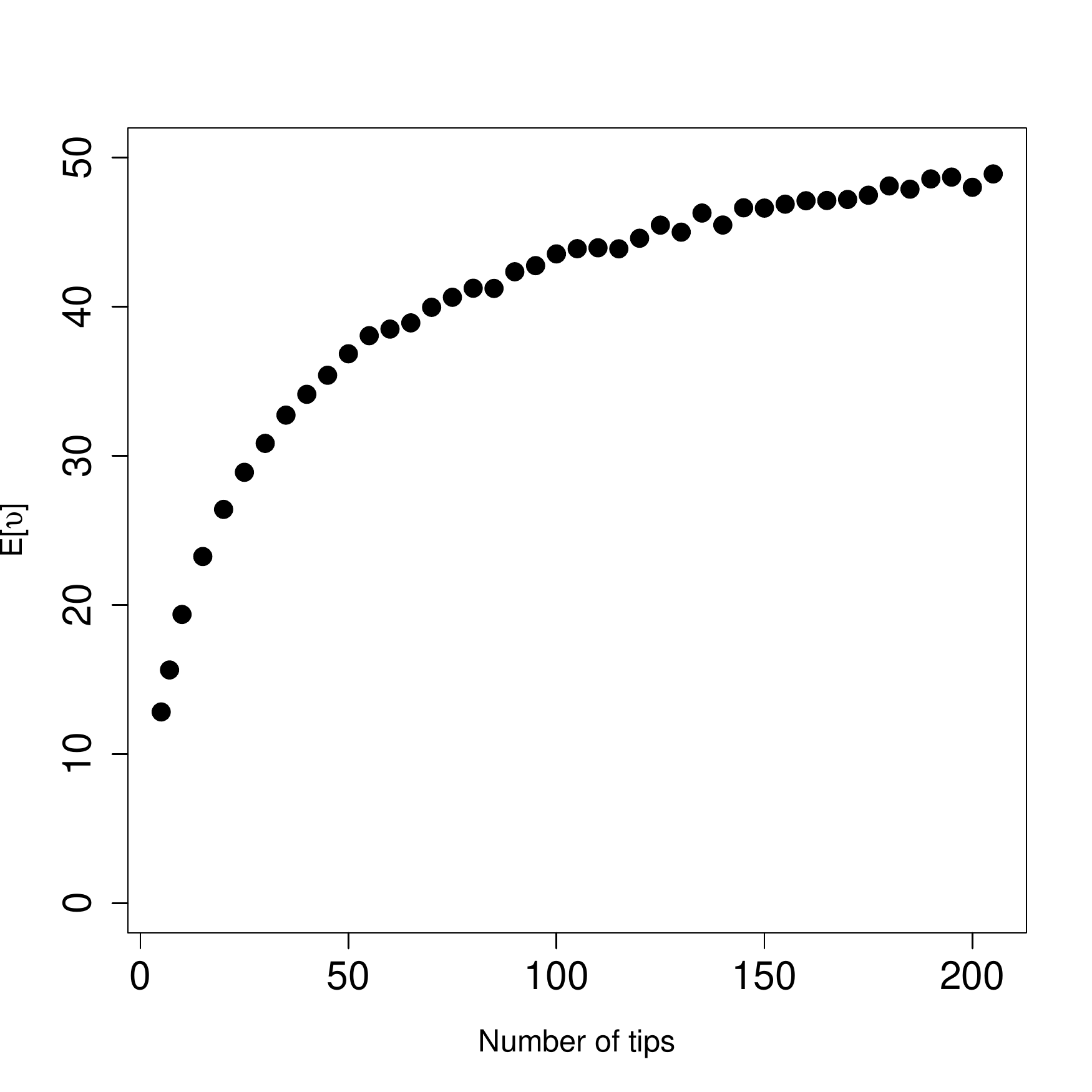} 
\caption{
Simulated $\E{\Upsilon}$ and $\E{\upsilon}$ as functions of $n$, the number of extant species, in the case of a supercritical 
birth--death tree
with $\lambda=0.134$ and $\mu=0.037$ top row ($10000$ simulated trees)
and $\lambda=0.46$, $\mu=0.43$ bottom row ($10000$ simulated trees). Simulations were done using the TreeSim \citep{TreeSim1,TreeSim2} R package.
}\label{figBDNunu}
\end{figure}

\begin{table}[!ht]
\centering
\begin{tabular}{cccc}
\hline
Parameter & Phylogeny $1$  & Phylogeny $2$  & DLMH \\
& \citep{APur1995} & \citep{FBokVBriTSta2012} &  \citep{FBokVBriTSta2012} \\
\hline
$n$ & $5$ & $7$ & --- \\
$\lambda$ & $0.134$ & $0.46$ & --- \\
$\mu$ & $0.037$ & $0.43$ & --- \\
$\E{\Upsilon}$ & $3.553$ & $20.952$ & $3.6$ \\
$\E{\upsilon}$ & $1.448$ & $15.639$ & --- \\
$\E{T}$ & $20.831$ & $24.935$ & $5.25$ \\
$\E{T-\tau}$ & $12.941$ & $19.182$ & --- \\
$\rho_{n}$ only gradual change & $0.621$ & $0.769$ & ---\\
$\rho_{n}$ gradual and punctuated change & $0.615$ & $0.766$ & --- \\
$\E{[X]}$ & $1.038$ & $1.384$ & $0.2845$ \\
$\%$ anagenetic & $97.11\%$ & $87.21\%$ & $89.31\%$ \\
$\%$ cladogenetic & $2.89\%$ & $12.79\%$ & $10.39\%$ \\
\hline
\end{tabular}
\caption{Summary of parameters and calculated values for the Hominoid body size analysis.
} \label{tabResults}
\end{table}

\begin{table}[!ht]
\centering
\begin{tabular}{cc}
\hline
Species & body size [$\ln(kg)$] \\ 
\hline
\textit{Gorilla gorilla} & $4.74$ \\
\textit{Homo sapiens} &  $4.1$ \\
\textit{Pongo Pygmaeus} & $4.02$ \\
\textit{Pan paniscus} & $3.66$ \\
\textit{Pan troglodytes} & $3.62$ \\
\hline
\end{tabular}
\caption{Hominoidea body size measurements, read off Fig 2. in \citep{FBok2002}.} \label{tabHomBodySize}
\end{table}

In his study \citet{FBok2002} estimated the parameters of anagenetic and cladogenetic change and concluded
that one cannot reject the null hypothesis that body size evolution in Hominoidea has been entirely gradual.
I will assume here his estimated values of $\sigma_{a}^{2}$ and $\sigma_{c}^{2}$ and conditional on them
see how punctuated change effects Hominoidea. 
Therefore some way of working with the number of jumps that occurred is required
and there are two possible ways to do this. One is to use the estimated values of $\E{\Upsilon}$
and $\E{\upsilon}$. However, since I have a tree available I may try to estimate the number of jumps 
conditional on the tree for a particular lineage of interest. In this case there is no need
to use $\E{T}$ as the lineage length will be known. I will discuss both approaches and in the second case
choose after \citet{FBokVBriTSta2012} the direct line to modern human (DLMH) as the lineage of interest.


From Tab. \ref{tabResults}
One can see that in the case of the first five species phylogeny 
cladogenetic change
decreases the species' similarity by $0.0062$ (or by $1\%$). 
To get some sort of comparison with the real data (Fig. \ref{tabHomBodySize})
I considered all $20$ pairs of species 
and from this estimated the correlation as $\hat{\rho_{n}}=0.43$. Notice that this is not entirely correct
as the sampled pairs are statistically dependent. However, this empirical correlation coefficient is on the same
level as the theoretical one of $0.6151$.
One can also see that the punctuated change consists of $2.9\%$ of the expected change of the phenotype.
Therefore even if the estimate of $\sigma_{a}^{2}$ from \citet{FBok2002} was significant the role of punctuated change
would not be that crucial.

Next I studied evolution on the DLMH. Its length is suggested to be between $4.1$ and $7.02$ Mya
\citep{AHobetal2007,SKumetal2005,NPatetal2006}.
I chose $T\approx 5.25$[Mya] \citep[][phylogeny in supporting material]{FBokVBriTSta2012}
to be consistent with the estimated number of speciation events 
\citep[even though in][according to which $\sigma_{c}^{2}$ was estimated,
7.04Mya is stated, this however would only increase the contribution of gradual change by about $2.5\%$]{APur1995}.

The estimated expected number of speciation events on the lineage is $3.6$ \citep{FBokVBriTSta2012}. 
This gives the expected amount of log--body--mass change in the lineage to \textit{Homo sapiens} from our divergence
from chimpanzees to be $0.2845$ and in \mbox{Tab. \ref{tabResults}}
one see that gradual change remains the dominant one but the punctuated change plays some role as well. 
The two \textit{Pan}
species have an average log--body--mass of $3.64\ln(kg)$ while \textit{Homo sapiens} has $4.1\ln(kg)$. 
The difference between
these two values is $0.46$, not very far from $\sqrt{2\cdot 0.2845} \approx 0.754$, the square root 
of the expected total amount of change,
occurring on the distance from \textit{Homo sapiens} to the \textit{Pan} clade. 
We take the square root to have the same unit of measurement, $\ln(kg)$, in both cases.
Actually one would expect the actual difference
to be lower than the total change. This is because the real fluctuations are random and will not always be divergent.
Sometimes they
will be convergent, i.e. the change in both lineages will be sometimes in the same and sometimes in a different direction, therefore
some changes will cancel out. The ``total change'' (or quadratic variation) on the path connecting 
the two clades on the other hand, will not have this evening out effect, 
but will first make all change divergent and then sum it. 

Comparing the percentage of punctuated change to the case with the random birth--death tree
one can see that even though both indicate that punctuated change is not the dominating explanation for evolution,
they differ in the estimated magnitude of it. I will compare what happens if instead of the birth--death rates
of \citet{APurSNeePHar1995} I use $\lambda=0.46$, $\mu=0.43$
and $7$ species \citep[after][]{FBokVBriTSta2012}. 
In \mbox{Tab. \ref{tabResults}} one can see higher correlation coefficients 
and
the proportion of anagenetic and cladogenetic change is similar to what I obtained
from the 
DLMH calculations. 
This however should not be that surprising as 
both the length of the DLMH and the birth and death rates used above were jointly estimated
by \citet{FBokVBriTSta2012}.

With this analysis I do not intend to say anything definitive about Hominoid body size evolution. It is
only 
meant to serve as an illustration for the methodology presented here. All of the calculations seem to indicate 
that evolution in this clade was predominantly driven by gradual change.

\section{Discussion}
\label{secDis} 
\citet{FBok2002} introduced in his work 
\citep[but see also][]{FBok2003,FBok2008,FBok2010,TMatFBok2008,AMooDSch1998,AMooSVamDSch1999}
a modelling approach allowing for both gradual change and punctuated equilibrium. The 
modelling framework presented here is compatible with his but looks at it from a different perspective.
\citet{FBok2002} was interested in detecting punctuated equilibrium and devised a statistical procedure
for it, including confidence intervals. I ask another question: if I know 
(have estimates of) rates of cladogenetic and anagenetic evolution, can I then quantify (as e.g. 
fractions/percentages) how much of evolutionary change is due to gradual and how much due to cladogenetic change. 
These are two  different questions. The first is a statistical one, is there enough
data to estimate a parameter and how to do it. The second is descriptive, what do parameters
mean and what effects do they have on the system under study
\citep[but see also][for quantifying effects of gradual and punctuated evolution]{TMatFBok2008}.

In Eq. \eqref{eqQVtriat} one can see that the magnitude of each type of evolution depends on two components,
its instantaneous effect ($\sigma_{a}^{2}$ and $\sigma_{c}^{2}$) and the time over which each
was allowed to act ($T$ and $\Upsilon^{\ast}$). If speciation is very rare, then even if $\sigma_{c}^{2}$
is much larger than $\sigma^{2}_{a}$, punctuated equilibrium will have a tiny effect on the final 
evolutionary outcome. However, if the jumps are frequent, cladogenetic evolution can have a very large effect
even if its magnitude is small each time it occurs.

I observed in addition that adding a cladogenetic component has the effect of making species less and less similar,
as expected from intuition. The magnitude of the cladogenetic effect could give some insights
into a trait's role. As suggested by \citet{AMooSVamDSch1999}
cladogenetic evolution could be linked to 
traits that are involved in speciation or niche shifts and gradual evolution
to traits under continuous selection pressures that often change direction.

I also believe that my characterization of the Ornstein--Uhlenbeck process with jumps is an
important contribution to evolutionary modelling. 
This model is consistent with many theories on trait evolution. 
It is constrained and so it is well suited to include the fact that the trait has to function with other 
traits and cannot be arbitrarily large or small as then it would become useless. Of course
some freedom to vary around the optimum has to be present, e.g. to allow the organism to adapt 
if the environment changes. This is exactly how the Ornstein--Uhlenbeck process behaves. At stationarity
it is made up of constant, finite variance oscillations around the optimum. If the optimum changes, 
the trait will be pulled towards the new optimum. If the phenotype changes radically due to e.g.
rapid speciational change, it will be pulled back towards the optimum. 

Most phylogenetic comparative data sets contain only contemporary
measurements. These measurements have accumulated both cladogenetic and
anagenetic effects over the whole course of their evolution. Even if the instantaneous effect of cladogenetic
evolution is large, if it occurred rarely compared to the time of anagenetic evolution 
it could be difficult to obtain statistical significance of parameter estimates. Therefore, it is crucial 
\citep[as also pointed out in][]{FBok2002} to include speciation events resulting in extinct lineages. 
At each hidden speciation event cladogenetic evolution may have taken place. 

From the Hominoid body size example discussed here
one can conclude
that in the study of punctuated equilibrium 
if one conditions on the speciation and extinction rates 
it is important to obtain correct values of 
them to have the right balance between the time when gradual change 
took place and the number of opportunities for punctuated change
especially with a small number of tip species. I conducted an analysis based
on two studies of essentially the same clade but with different birth and death rate estimates
and obtained different (though qualitatively similar) proportions.
However, one can also see that what enters Eq. \eqref{eqQvarJ} is the product $\sigma_{c}^{2}p\E{\Upsilon}$.
Therefore one is free to vary two components $\sigma_{c}^{2}$ and $p\E{\Upsilon}$ to obtain 
the cladogenetic contribution. A low jump frequency and a high variance of the jump will give the same
effect as a high frequency of low magnitude jumps. Therefore it seems plausible that in a joint estimation
procedure even if it would be difficult to reliably estimate the individual parameters
$\sigma_{c}^{2}$, $\mu$ and $\lambda$ one still might be able to obtain reasonable values
for the contribution of each mode of evolution.

I have presented analytical results for the conditioned Yule model. This model does not allow
for extinction but it is a starting point for integrating punctuated equilibrium models
of phenotypic evolution with models of tree growth. As I discussed, trees with extinction
can easily be handled by simulation methods and a pure birth model can serve as a 
convenient prior for a Bayesian approach or as the basis for a numerical procedure. 
In addition to this I have derived
the probability generating function for the number of speciation events on a random
lineage and for the number of speciation events from the most recent common ancestor
of two randomly sampled tip species. These are to the best of my knowledge
novel results \citep[but see also e.g.][for further distributional properties of Yule trees]{TreeSim1,MSteAMcK2001}.

\section*{Acknowledgments}
I would like to thank Folmer Bokma for his many helpful comments, corrections and reference suggestions
that immensely improved the manuscript. I would like to thank Olle Nerman 
for his suggestions leading to the study of models with jumps at speciation events,
and Thomas F. Hansen, Serik Sagitov and Anna Stokowska for many valuable comments and suggestions.
I was supported by the Centre for Theoretical Biology at the University of Gothenburg, 
Stiftelsen f\"or Vetenskaplig Forskning och Utbildning i Matematik
(Foundation for Scientific Research and Education in Mathematics), 
Knut and Alice Wallenbergs travel fund, Paul and Marie Berghaus fund, the Royal Swedish Academy of Sciences,
and Wilhelm and Martina Lundgrens research fund.
A previous version of this manuscript went into my PhD thesis at the Department of Mathematical Sciences,
University of Gothenburg, Gothenburg Sweden.

\appendix

\section{Theorem proofs}
\subsection{Proof of Theorem \ref{thPGF}}
\begin{proof}
$\Upsilon$ was written as,
\bd
\Upsilon = \sum_{k=1}^{n-1} \mathbf{1}_{k},
\ed
where the indicator random variables $\mathbf{1}_{k}$ represent whether the $k$--th
coalescent event is on the sampled lineage and by properties of the Yule tree they
are independent and equal one with probability $p_{k_{1}}=2/(k+1)$ ($p_{k_{0}}=1-p_{k_{1}}$). Therefore,
\bd
\begin{array}{l}
\E{s^{\Upsilon}} = \prod\limits_{k=1}^{n-1}\E{s^{\mathbf{1}_{k}}}
=
\prod\limits_{k=1}^{n-1}(p_{k_{0}}+sp_{k_{1}})
=\prod\limits_{k=1}^{n-1}\frac{k+2s-1}{k+1}
=\frac{1}{n}\frac{1}{b_{n-1,2s-1}},
\end{array}
\ed

\bd
\begin{array}{l}
\E{s^{\upsilon}} = \sum\limits_{k=1}^{n-1}\E{s^{\upsilon} \vert k}\pi_{k,n}
=
\pi_{1,n}+\sum\limits_{k=2}^{n-1}\pi_{k,n}\prod\limits_{i=1}^{k-1}\E{s^{\mathbf{1}_{i}}}
\\=
\pi_{1,n}+\sum\limits_{k=2}^{n-1}\pi_{k,n}\prod\limits_{i=1}^{k-1}\left(p_{i_{0}} + sp_{i_{1}} \right)
=\pi_{1,n}+2\frac{n+1}{(n-1)\Gamma(2s)}\sum\limits_{k=1}^{n-2}\frac{\Gamma(k+2s)}{\Gamma(k+3+1)}.
\end{array}
\ed
I assumed $n>2$ 
(if $n=2$ then by definition $\upsilon=0$ and so for all $s>0$ $\E{s^{\upsilon}}=1$) and I use the property
\citep[see also][]{KBarSSag2012}, 
\be\label{eqGsum}
\sum\limits_{k=1}^{n-1}\frac{\Gamma(k+y)}{\Gamma(k+z+1)} = 
\left\{
\begin{array}{ll}
\frac{\Gamma(n+z)\Gamma(y+1)-\Gamma(z+1)\Gamma(n+y)}{\Gamma(z+1)\Gamma(n+z)(z-y}, & z \neq y, \\
\Psi(n+y)-\Psi(1+y), & z=y,
\end{array}
\right.
\ee
where $\Psi(z)$ is the polygamma function defined as $\Psi(z)=\Gamma'(z)/\Gamma(z)$.
In the case of $z\neq y$ the formula above is verifiable by induction and if $z=y$ the formula can be either calculated
directly from the left side or as the limit of the right side $z\rightarrow y$.

This property and that $\pi_{1,n}=(n+1)/(3(n-1))$ gives when $s\neq1.5$,
\bd
\begin{array}{l}
\E{s^{\upsilon}} = 
\frac{1}{(n-1)(3-2s)}\left(
(n+1)-\frac{2}{nb_{n-1,2s-1}}
\right).
\end{array}
\ed

When $s=1.5$ I get,
\bd
\begin{array}{l}
\E{s^{\upsilon}} 
= \frac{2(n+1)}{n-1}\left(H_{n+1}-\frac{5}{3}\right).
\end{array}
\ed
\end{proof}

\subsection{Proof of Theorem \ref{thRhon}}
\begin{proof}
Let $\vec{t} = (t_{1},t_{2},\ldots,t_{\Upsilon},t_{\Upsilon+1})$ be the 
between speciation times on a randomly chosen lineage, see \mbox{Fig. \ref{figappTreeTimes}} for illustration.
For mathematical convenience I set the speciation rate $\lambda =1$
as changing $\lambda$ is equivalent to rescaling the branch lengths by its inverse
and does not effect the topology ($\Upsilon$ and $\upsilon$ here) because
I have conditioned on $n$. Therefore
a Yule--Ornstein--Uhlenbeck--jumps model for the extant species trait sample with parameters 
$(\alpha,\sigma_{a}^{2}, \sigma_{c}^{2}, \theta, X_{0},\lambda)$
is equivalent to one with parameters $(\alpha/\lambda, \sigma_{a}^{2}/\lambda, \sigma_{c}, \theta,  X_{0},1)$.
By $X^{+}_{t_{i}}$
I will denote the value of the process just after ($^{+}$ indicating I include the jump if it occurred) 
the node ending the branch corresponding
to duration $t_{i+1}$, 
notice that $X^{+}_{t_{\Upsilon+1}}$ will be the value just after an extent species
and $X^{+}_{t_{\upsilon+1}}$ will be the value just after the node being the most recent
common ancestor of two randomly sampled species. If I don't include the $^{+}$, i.e.
$X_{t_{i}}$ I mean the value of the process at the respective node.
As in the proof of Theorem \ref{thPGF} $\mathbf{1}_{k}$ 
is the $0$--$1$ random variable indicating 
whether the $k$--th ($k=1,\ldots, n-1$)
coalescent event is on the sampled lineage, equalling $1$ with probability $2/(k+1)$.
\begin{figure}[!h]
\centering
\includegraphics[width=5cm]{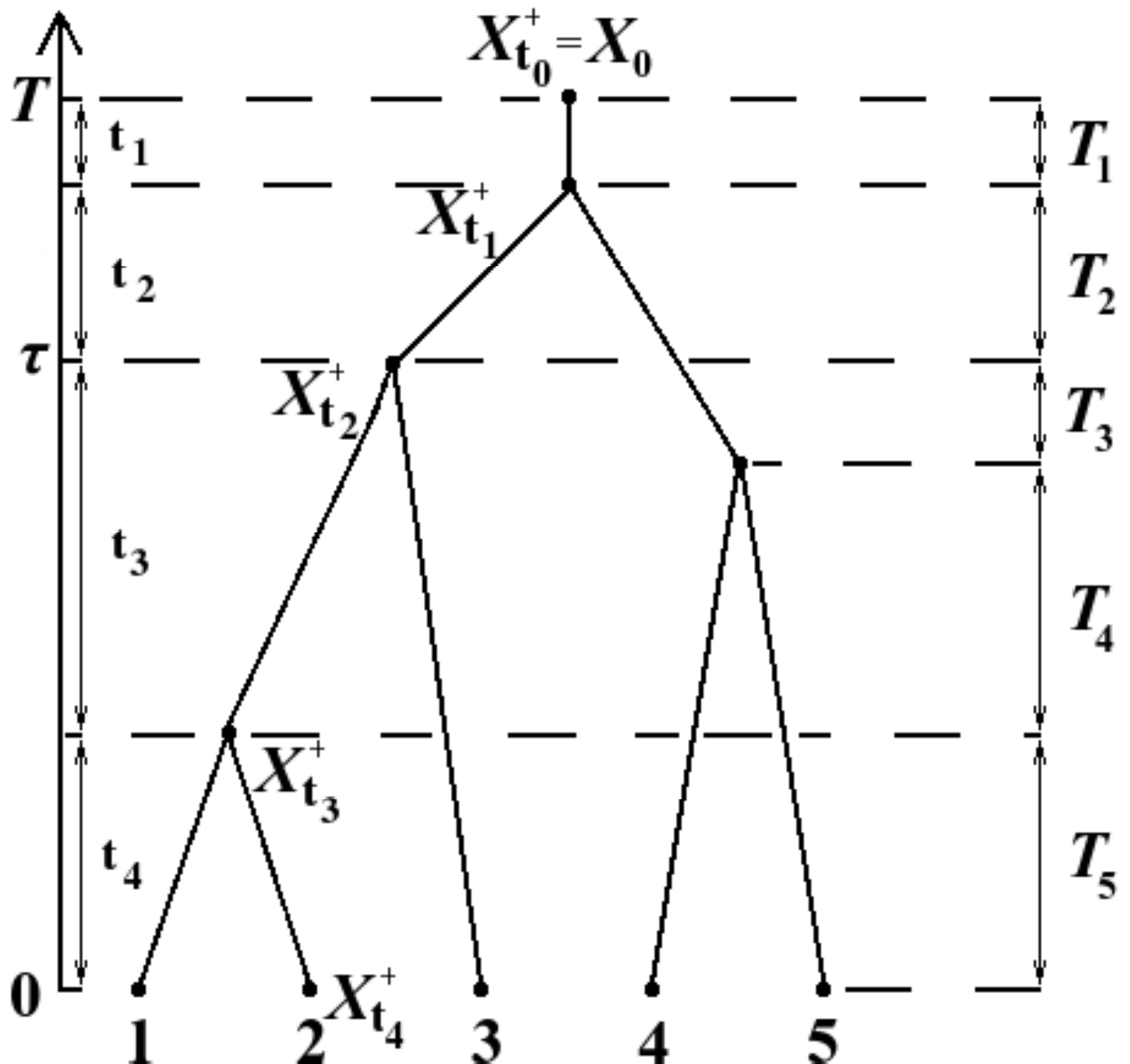}
\caption{
A pure--birth tree with the various time components marked on it. 
If I ``randomly sample'' node two then I will have $\Upsilon=3$ and $\vec{t}=(t_{1},t_{2},t_{3},t_{4})$.
I have $t_{1}=T_{1}$, $t_{2}=T_{2}$, $t_{3}=T_{3}+T_{4}$ and $t_{4}=T_{5}$.
The process values just after the nodes are $X^{+}_{t_{0}},X^{+}_{t_{1}},X^{+}_{t_{2}},X^{+}_{t_{3}},X^{+}_{t_{4}}$,
with $X^{+}_{t_{4}}$ equalling the extant value $X_{2}$ if node $2$ did not split or if node $2$ split
$X^{+}_{t_{4}}$ would be $X_{2}$ plus a jump (if one took place). If I ``randomly sample'' the pair of extant 
species $2$ and $3$ then $\upsilon=1$ and the two nodes coalesced when the process was just 
before $X^{+}_{t_{2}}$.
}
\label{figappTreeTimes}
\end{figure}
The expectation of an extant species is,
\bd
\begin{array}{l}
\E{X} = 
\E{\E{X\vert X_{t_{\Upsilon}},\vec{t}~}} 
=
\E{\E{(X_{t_{\Upsilon}}-\theta)e^{-\alpha t_{\Upsilon+1}}+\theta \vert \vec{t}~}}
\\=
\E{e^{-\alpha t_{\Upsilon+1}}(\E{X_{t_{\Upsilon}}\vert \vec{t}~}-\theta)}+\theta
\\=
\E{e^{-\alpha T}(X_{0}-\theta)} + \theta = b_{n,\alpha}X_{0} + (1-b_{n,\alpha})\theta.
\end{array}
\ed
I now turn to the variance,
\bd
\begin{array}{l}
\var{X} = \E{\var{X \vert \vec{t}~}} + \var{\E{X \vert \vec{t}~}} .
\end{array}
\ed
I consider the most complicated term,
\bd
\begin{array}{l}
\var{X \vert \vec{t}~}
=
\E{\var{X \vert \vec{t},X^{+}_{t_{\Upsilon}}}} + \var{\E{X \vert \vec{t},X^{+}_{t_{\Upsilon}}}}
\\=
\frac{\sigma_{a}^{2}}{2\alpha}(1-e^{-2\alpha t_{\Upsilon+1}}) + e^{-2\alpha t_{\Upsilon+1}}\var{X^{+}_{t_{\Upsilon}} \vert \vec{t}~}
\\=
\frac{\sigma_{a}^{2}}{2\alpha}(1-e^{-2\alpha(t_{k+1}+\ldots+t_{\Upsilon+1})})
+ 
e^{-2\alpha(t_{k+1}+\ldots+t_{\Upsilon+1})}\var{X^{+}_{t_{k}} \vert \vec{t}~}
+ \sigma_{c}^{2}p
\sum\limits_{i=k+2}^{\Upsilon+1}e^{-2\alpha (t_{i}+\ldots+t_{\Upsilon+1})}
\\=
\frac{\sigma_{a}^{2}}{2\alpha}(1-e^{-2\alpha T}) + 
\sigma_{c}^{2}p\sum\limits_{i=2}^{\Upsilon+1}e^{-2\alpha (t_{i}+\ldots+t_{\Upsilon+1})}.
\end{array}
\ed
I will use the equality, for $y>0$, obtainable from Eq. \eqref{eqGsum},
\bd
\begin{array}{rcl}
\sum\limits_{i=1}^{n-1}\frac{\Gamma(i+y+1)}{\Gamma(i+2)}&  = & \frac{1}{y}\left(\frac{\Gamma(n+y+1)}{\Gamma(n+1)}-\Gamma(y+2) \right),
\end{array}
\ed
to calculate,
\bd
\begin{array}{l}
\E{\sum\limits_{i=2}^{\Upsilon+1}e^{-2\alpha (t_{\Upsilon+1}+\ldots+t_{i})}}
=
\E{\sum\limits_{i=1}^{n-1}\mathbf{1}_{i}e^{-2\alpha (T_{n}+\ldots+T_{i+1})}}
=
2\sum\limits_{i=1}^{n-1}\frac{b_{n,2\alpha}}{(i+1)b_{i,2\alpha}}
\\=
\frac{2b_{n,2\alpha}}{\Gamma(2\alpha+1)}\sum\limits_{i=1}^{n-1}\frac{\Gamma(i+2\alpha+1)}{\Gamma(i+2)}
= \frac{2}{2\alpha}\left(1-(1+2\alpha)b_{n,2\alpha} \right).
\end{array}
\ed
Using the above and 
\bd
\var{\E{X \vert \vec{t}~}} = (X_{0}-\theta)^{2}\var{e^{-2\alpha T}},
\ed
I arrive at,
\bd
\begin{array}{l}
\var{X} = 
\frac{\sigma_{a}^{2}}{2\alpha}(1-\E{e^{-2\alpha T}}) + (X_{0}-\theta)^{2}\var{e^{-\alpha T}} +
\frac{2p\sigma_{c}^{2}}{2\alpha}
\left(1-(1+2\alpha)b_{n,2\alpha} \right).
\end{array}
\ed
Writing $\kappa=2p\sigma_{c}^{2}/(\sigma_{a}^{2}+2p\sigma_{c}^{2})$, $\delta=\vert X_{0}-\theta \vert/\sqrt{\sigma_{a}^{2}/2\alpha}$,
\be
\begin{array}{rcl}
\var{X} & =  &
\frac{\sigma_{a}^{2}+2p\sigma_{c}^{2}}{2\alpha}\left(
(1-\kappa)\left(1-b_{n,2\alpha} + \delta^{2}(b_{n,2\alpha}-b_{n,\alpha}^{2})\right) 
\right.\\ && \left.
+ \kappa \left(1-(1+2\alpha)b_{n,2\alpha} \right)
\right).
\end{array}
\ee
In the same manner I will calculate the covariance between two randomly sampled tip species
(with $\vec{t}_{1}, \vec{t}_{2}$ being the vectors of
between speciation times on the respective species' lineage),
\bd
\begin{array}{l}
\cov{X_{1}}{X_{2}} = 
\cov{\E{X_{1} \vert \vec{t}_{1},\vec{t}_{2},X_{t_{\upsilon+1}},\upsilon}}{\E{X_{2} \vert \vec{t}_{1},\vec{t}_{2},X_{t_{\upsilon+1}},\upsilon}}
\\ + \E{\cov{X_{1}}{X_{2} \vert \vec{t}_{1},\vec{t}_{2},X_{t_{\upsilon+1}},\upsilon}} 
= 
\var{e^{-\alpha \tau}(X_{t_{\upsilon+1}}-\theta)}
\\ = 
\E{\var{e^{-\alpha \tau}(X_{t_{\upsilon+1}}-\theta)\vert \vec{t}_{1},\vec{t}_{2},X^{+}_{t_{\upsilon}},\upsilon}} + \var{\E{e^{-\alpha \tau}(X_{t_{\upsilon+1}}-\theta)\vert \vec{t}_{1},\vec{t}_{2},X^{+}_{t_{\upsilon}},\upsilon}}
\\ =
\frac{\sigma_{a}^{2}}{2\alpha}\E{e^{-2\alpha\tau}-e^{-2\alpha(\tau+t_{\upsilon+1})}} 
+ \var{e^{-\alpha(\tau+t_{\upsilon+1})}(X^{+}_{t_{\upsilon}}-\theta)}
\\=
\frac{\sigma_{a}^{2}}{2\alpha}\E{e^{-2\alpha\tau}-e^{-2\alpha(\tau+t_{\upsilon+1}+t_{\upsilon})}} 
+ \var{e^{-\alpha(\tau+t_{\upsilon+1}+t_{\upsilon})}(X^{+}_{\upsilon-1}-\theta)}
+p\sigma_{c}^{2}\E{e^{-2\alpha(\tau+t_{\upsilon+1})}}
\\=
\frac{\sigma_{a}^{2}}{2\alpha}\left(\E{e^{-2\alpha \tau}}-\E{e^{-2\alpha T}} \right)+(X_{0}-\theta)^{2}\var{e^{-\alpha T}}+p\sigma_{c}^{2}\E{\sum\limits_{i=2}^{\upsilon+1}e^{-2\alpha(\tau+t_{\upsilon+1}+\ldots+t_{i})}}.
\end{array}
\ed
Let $\kappa$ denote number of the speciation event when the two randomly sampled tips coalesced 
 and I consider for $y>0$,
\bd
\begin{array}{l}
\E{\sum\limits_{i=2}^{\upsilon+1}e^{-y(\tau+t_{\upsilon+1}+\ldots+t_{i})}} 
= \sum\limits_{k=1}^{n-1}\pi_{k,n}\sum\limits_{i=1}^{k-1}\E{\mathbf{1}_{i}e^{-y(T_{n}+\ldots+T_{i+1})} \vert \kappa = k}
\\
=\frac{4(n+1)b_{n,y}}{(n-1)}\sum\limits_{k=1}^{n-1}\sum\limits_{i=1}^{k-1}\frac{1}{(k+1)(k+2)(i+1)b_{i,y}}
=\frac{4(n+1)b_{n,y}}{y\Gamma(y)(n-1)}\sum\limits_{k=1}^{n-1}\frac{1}{(k+1)(k+2)}\sum\limits_{i=1}^{k-1}\frac{\Gamma(i+1+y)}{\Gamma(i+2)}
\\
=\frac{4(n+1)b_{n,y}}{y^{2}\Gamma(y)(n-1)}\sum\limits_{k=1}^{n-1}\frac{1}{(k+1)(k+2)}\left(\frac{\Gamma(k+y+1)}{\Gamma(k+1)}-\Gamma(y+2) \right),
\end{array}
\ed
equalling by Eq. \eqref{eqGsum} for $0<y\neq 1$
\bd
\frac{2}{y}\left(\frac{2-(y+1)(yn-y+2)b_{n,y}}{(n-1)(y-1)}\right)
\ed
and for $y=1$ this is,
\bd
2\frac{2}{n-1}\left(H_{n} -\frac{5n-1}{2(n+1)} \right).
\ed
Putting all of the above together I obtain,
\bd
\begin{array}{rcl}
\cov{X_{1}}{X_{2}} & =  &
\left\{
\begin{array}{l}
\frac{\sigma_{a}^{2}}{2\alpha}\left(\frac{2-(n+1)(2\alpha+1)b_{n,2\alpha}}{(n-1)(2\alpha-1)} - b_{n,2\alpha} \right)
+(X_{0}-\theta)^{2}\left(b_{n,2\alpha} - b_{n,\alpha}^{2} \right)  
\\
+\frac{2p\sigma_{c}^{2}}{2\alpha}\frac{2-(2\alpha n -2\alpha +2)(2\alpha+1)b_{n,2\alpha}}{(n-1)(2\alpha-1)},~~~0<\alpha\neq 0.5, \\

\frac{\sigma_{a}^{2}}{2\alpha}\left(\frac{2}{n-1}\left(H_{n}-1 \right) - \frac{2}{n+1} \right)
+(X_{0}-\theta)^{2}\left(\frac{1}{n+1} - b_{n,0.5}^{2} \right) 
\\
+ 2p\sigma_{c}^{2} \frac{2}{n-1}\left(H_{n}  - \frac{5n-1}{2(n+1)} \right),~~~\alpha = 0.5 
\end{array}
\right.
\end{array}
\ed
and with $\kappa$ and $\delta$,
\be
\begin{array}{rcl}
\cov{X_{1}}{X_{2}} & =  &
\frac{\sigma_{a}^{2}+2p\sigma_{c}^{2}}{2\alpha}
\left\{
\begin{array}{l}
(1-\kappa)\left(\frac{2-(n+1)(2\alpha+1)b_{n,2\alpha}}{(n-1)(2\alpha-1)} - b_{n,2\alpha} 
+\delta^{2}\left(b_{n,2\alpha} - b_{n,\alpha}^{2} \right) \right)
 
\\
+\kappa\frac{2-(2\alpha n -2\alpha +2)(2\alpha+1)b_{n,2\alpha}}{(n-1)(2\alpha-1)},~~~ 0<\alpha\neq 0.5, \\

(1-\kappa)\left(\frac{2}{n-1}\left(H_{n}-1 \right) - \frac{2}{n+1}
+\delta^{2}\left(\frac{1}{n+1} - b_{n,0.5}^{2} \right) \right)
  
\\
+ \kappa \frac{2}{n-1}\left(H_{n}  - \frac{5n-1}{2(n+1)} \right),~~~\alpha = 0.5. 
\end{array}
\right.
\end{array}
\ee
\end{proof}

\section{Quadratic variation}
In order to quantify how much a function (for example a trait as a function of time evolving on a lineage) 
changed over time one needs to have some measure of change. 
It is not sufficient to compare the current function (trait) value (or average of trait values) with that
of the  value of the function at the origin (ancestral trait). It can happen that the function might have over
time made large excursions but just by chance is currently at a value close to the origin, especially if the expectation of the (random) function
equals that of the original value, like in the Brownian motion (neutral evolution) one.

One therefore needs some measure that will consider not only how the function compares to what it
was at some time in the past but will also (or rather primarily) take into account how diverse its path
was. One such possible measure is the quadratic variation of a function. 
\begin{definition*}
The quadratic variation of a function $f(t)$ is defined as,
\bd
[f](t) = \lim\limits_{\max\limits_{1\le i \le n}t_{i}-t_{i-1}\searrow 0}\sum_{i=1}^{n}(f(t_{i})-f(t_{i-1}))^{2},
\ed
where $0=t_{0}<t_{1}<\ldots <t_{n}=t$ are finer and finer partitions of the interval $[0,t]$.
\end{definition*}
The advantage of using quadratic variation to quantify the phenotype's change is that it is mathematically easy
to work with and will be consistent with the ideas and results of \citet{FBok2002,TMatFBok2008}.
The quadratic variation is a property of a function and so can be applied to a trajectory of a stochastic process.
In fact one can show \citep[see e.g.][]{FKle} that the quadratic variation of a 
stochastic process defined by an SDE of the form,
\begin{equation}\label{eqappSDEdiff}
\ud X(t) = \mu(t,X(t))\ud t + \sigma \ud B(t)
\end{equation}
where $B(t)$ is a standard Brownian motion and with appropriate conditions on $\mu(\cdot,\cdot)$
will be $[X](t)=\sigma^{2}t$. This is rather remarkable as the quadratic variation is a property 
of the process' trajectory (a random function) but it turns out to be a deterministic function.
From the definition one can see that the quadratic variation takes into account all the local
fluctuations of the phenotype and then sums them up to get the accumulated fluctuation.
The considered above family of SDEs is sufficient for my purposes as it encompasses
the two most widely used evolutionary models, Brownian motion, $\mu \equiv 0$, and Ornstein--Uhlenbeck
process, $\mu(t,X(t))=-\alpha(X(t)-\theta)$.
For further properties of the quadratic variation the reader is referred to stochastic calculus
literature \citep[e.g.][]{FKle,PMed,BOxe}.

\bibliographystyle{plainnat}
\bibliography{SagitovBartoszek}

\end{document}